\documentclass[11pt]{article}
\usepackage[a4paper, margin=1in]{geometry}
\usepackage{amsmath,amssymb,amsthm}
\usepackage{color}
\usepackage{authblk}
\usepackage{graphicx}
\usepackage{hyperref}
\usepackage[capitalize,nameinlink]{cleveref}
\usepackage{crossreftools}

\newtheorem{theorem}{Theorem}
\newtheorem{lemma}[theorem]{Lemma}
\newtheorem{claim}[theorem]{Claim}
\newtheorem{subclaim}[theorem]{Subclaim}
\newtheorem{definition}[theorem]{Definition}

\newenvironment{subproof}[1][\proofname]{%
	\begin{proof}[#1]%
	}{%
	\end{proof}%
}

\usepackage{float}
\floatstyle{ruled}
\newfloat{algorithm}{tbp}{loa}
\providecommand{\algorithmname}{Algorithm}
\floatname{algorithm}{\protect\algorithmname}

\begin{document}
	
	\title{Bellman-Ford in Almost-Linear Time for Dense Graphs}
	
	\author[1]{George Z. Li}
	\author[1]{Jason Li}
	\author[2]{Junkai Zhang}
	\affil[1]{Carnegie Mellon University $\{\texttt{gzli}, \texttt{jmli}\}$\texttt{@cs.cmu.edu}}
	\affil[2]{Tsinghua University \texttt{zhangjk22@mails.tsinghua.edu.cn}}
	
	\date{}
	\maketitle
	
	\begin{abstract}
		We consider the single-source shortest paths problem on a directed graph with real-valued (possibly negative) edge weights and solve this problem in $n^{2+o(1)}$ time by refining the shortcutting procedure introduced in Li, Li, Rao, and Zhang (2026).
	\end{abstract}

	\section{Introduction}

	We consider the problem of computing single-source shortest paths on weighted directed graphs. Despite its long history, the past few years have witnessed exciting breakthroughs for this problem in multiple different settings. When all edge weights are non-negative, the classic Dijkstra's algorithm solves the problem in $O(m+n\log n)$ time, which has been improved for sparse graphs to $O(m\log^{2/3}n)$~\cite{duan2025breaking}. When edge weights are integral (but possibly negative), the problem was solved in $\tilde O(m)$ time\footnote{$\tilde O(\cdot)$ ignores factors poly-logarithmic in $n$.}~\cite{bernstein2025negative} and $m^{1+o(1)}$ time~\cite{chen2025maximum}, with the latter solving the more general problem of minimum-cost flow. 
	
	In the general case of real-valued edge weights, the 70-year-old Bellman-Ford algorithm was the fastest until the recent breakthrough $\tilde{O}(mn^{8/9})$ time algorithm of Fineman~\cite{fineman2024single}. Subsequent work has refined Fineman's approach~\cite{huang2025faster,huang2026faster,quanrud2025sparsification}, improving the runtime to $\tilde{O}(mn^{0.696}+(mn)^{0.850})$. Li, Li, Rao, and Zhang \cite{li2025shortcutting} recently introduced a new approach for the problem. Their algorithm adds shortcut edges to the graph to reduce the number of negative edges along shortest paths by a constant factor. By iterating this procedure, they can compute the shortest paths in the shortcutted graph using an $O(1)$-hop shortest path algorithm.
	
	The bottleneck of the algorithm in \cite{li2025shortcutting} was the $\tilde{O}(n^{2.5})$ time required to construct the shortcut edges added in each iteration. We refine their shortcut construction to obtain our $n^{2+o(1)}$ algorithm\footnote{We very recently learned that a similar result has been independently obtained by Sanjeev Khanna and Junkai Song. Their work appears on arXiv simultaneously with ours.}. Along the way, we develop a stronger version of the betweenness reduction technique used in all prior works on the topic, which we believe may be of independent interest.

	\begin{theorem}\label{thm:main}
		There is a randomized algorithm for single-source shortest paths on real-weighted directed graphs which runs in $n^{2+o(1)}$ time. 
	\end{theorem}

	\section{Preliminaries}
	
	All graphs $G=(V,E)$ in this paper are directed and have (possibly negative) real edge weights. If a vertex $u\in V$ has negative outgoing edges, we say $u$ is a \emph{negative vertex}, and let $N$ denote the set of negative vertices. For a vertex $u\in V$, we define $N^{out}(u)=\{v\in V:(u,v)\in E\}$ and $N^{in}(u)=\{v\in V:(v,u)\in E\}$ as the out-neighbors and in-neighbors of $u$. Finally, we let the distance between a set and a point be the minimum distance, i.e., $d(S,v)=\min_{s\in S}d(s,v)$.
	
	Like prior work, we heavily rely on a subroutine for efficiently computing single-source shortest paths with few negative edges. For a positive integer $h$ called the \emph{hop bound}, define the \emph{$h$-negative hop distance} $d^h(s,t)$ as the minimum weight $(s,t)$-path with at most $h$ negative edges. Given a source vertex $s$, we can compute $d^h(s,t)$ for all other $t\in V$ in $O(h(m+n\log n))$ using a hybrid of Dijkstra and Bellman-Ford~\cite{dinitz2017hybrid}.
	
	Also like prior work, our algorithm makes use of \emph{potential functions} $\phi:V\to\mathbb R$ that reweight the graph, where each edge $(u,v)$ has new weight $w_\phi(u,v)=w(u,v)+\phi(u)-\phi(v)$. The key property of potential functions is that for any vertices $s,t\in V$, all $(s,t)$-paths have their weight shifted by the same additive $\phi(s)-\phi(t)$. A potential $\phi$ is \emph{valid} if $w_\phi(u,v)\ge0$ for all edges $(u,v)$ with $w(u,v)\ge0$, i.e., it does not introduce any new negative edges.
	
	\subsection{Copies of Vertices}\label{sec:copy}

	In our algorithm, we will build graphs $G_{\tau}$ from the original input graph $G=(V,E)$, by adding vertices and edges. Vertices in $G_{\tau}$ are classified into two categories: a base vertex and a copy of some base vertex. The set of base vertices $V_0$ is initialized as $V_0=V$ and the set of copy vertices is initially empty. Throughout the algorithm, we will add new base vertices and copies of base vertices. Suppose we have $c$ copies of each base vertex $v\in V_0$. The original vertex is denoted $v_0$ and the copies are labeled $v_{i}$ for $i\in[c]$. We maintain the following invariants:
		\begin{enumerate}
		\item[{\crtcrossreflabel{(I1)}[item:graph-invariant-1]}] each copy $v_i$ is associated with a shift $\delta(v_i)\in\mathbb{R}$, with $\delta(v_0)=0$ for all $v\in V_0$.
		\item[{\crtcrossreflabel{(I2)}[item:graph-invariant-2]}] every edge $(u_i,v_j)\in G_{\tau}$ satisfies $\omega(u_i,v_j)\ge d_{G_{\tau}}(u_0,v_0)+\delta(v_j)-\delta(u_i)$
	\end{enumerate}
	These invariants formalize what it means for an added vertex to be a copy of a base vertex. The shift $\delta(v_i)$ represents how much ``higher'' the copy $v_i$ is relative to the base vertex $v_0$, where distances to higher vertices are larger. This is formalized in the constraints on the edge weights in \ref{item:graph-invariant-2}, which also guarantees that shortest path distances aren't decreased by adding copies.

	When we add a base vertex to the graph, we will prove that invariants are maintained. When we add a copy of base vertices, we will define a shift and apply the following claim:
	
	\begin{claim}\label[claim]{cl:add-copy}
		Consider graph $G_{\tau}$ satisfying \ref{item:graph-invariant-1}--\ref{item:graph-invariant-2}. For any base vertex $v\in V_0$, suppose one adds a new copy $v_{new}$ of $v$ with  shift $\delta(v_{new})\in\mathbb{R}$ to $G_{\tau}$, along with the following edges:
		\begin{itemize}
			\item For each $u\in V$, let $C(u,v)=\{v_i:\omega(u,v_i)+\delta(v_{new})-\delta(v_i)\ge 0\}$. If $C(u,v)\neq\emptyset$, add edge $(u,v_{new})$ of weight $\omega(u,v_{new})=\min_{v_i\in C(u,v)}\omega(u,v_i)+\delta(v_{new})-\delta(v_i)$.
			\item For each $w\in V$, let $C(v,w)=\{v_i:\omega(v_i,w)+\delta(v_i)-\delta(v_{new})\ge 0\}$. If $C(v,w)\neq\emptyset$, add edge $(v_{new},w)$ of weight $\omega(v_{new},w)=\min_{v_i\in C(v,w)}\omega(v_i,w)+\delta(v_i)-\delta(v_{new})$.
		\end{itemize}
		Then invariants \ref{item:graph-invariant-1} and \ref{item:graph-invariant-2} are still satisfied in the resulting graph.
	\end{claim}
    \begin{proof}
        Invariant \ref{item:graph-invariant-1} is still satisfied since the new vertex $v_{new}$ has an associated shift $\delta(v_{new})$. Next, we need to verify Invariant \ref{item:graph-invariant-2} for edges incident to $v_{new}$. For an in-edge $(u_i,v_{new})$, the weight is $\omega(u_i,v_{new})=\omega(u_i,v_j)+\delta(v_{new})-\delta(v_j)$ for some $v_j\in C(u_i,v)$ by definition in the claim statement. But by Invariant \ref{item:graph-invariant-2}, we know $w(u_i,v_j)\ge d_{G_{\tau}}(u_0,v_0)+\delta(v_j)-\delta(u_i)$. Thus, we have $\omega(u_i,v_{new})\ge d_{G_{\tau}}(u_0,v_0)+\delta(v_{new})-\delta(u_i)$, as required by Invariant \ref{item:graph-invariant-2}. The proof for out-edges $(v_{new},w_j)$ is symmetric.
    \end{proof}

	\subsection{Preprocessing: Unique Incident Negative Edges}
	
	We will preprocess the graph using the following routine, so that each vertex $u$ has at most one negative edge incident to $u$. This preprocessing step is implicit in \cite{fineman2024single}, and our construction is the exact same as the one given in Lemma 4 of \cite{li2025shortcutting}.
	\begin{lemma}\label[lemma]{lem:one-negative-outgoing-edge}
		Given a graph $G=(V,E)$ with $k$ negative vertices, there is a linear-time algorithm that outputs a graph $G'=(V',E')$ with $V'\supseteq V$ such that
		\begin{enumerate}
			\item For any $s,t\in V$ and hop bound $h$, we have $d_{G'}^h(s,t)=d_G^h(s,t)$,
			\item Each vertex $u$ has at most one negative edge incident to $u$,
			\item $G'$ also has $k$ negative vertices, and $|V'|=|V|+k$.
		\end{enumerate}
	\end{lemma}
	\begin{proof}
		For each negative vertex $u$, consider its outgoing edges $(u,v_1),\ldots,(u,v_\ell)$, and suppose that $(u,v_1)$ is the edge of smallest weight, with ties broken arbitrarily. Create a new vertex $u'$, add the edge $(u,u')$ of weight $(u,v_1)$, and replace each edge $(u,v_i)$ by the edge $(u',v_i)$ of weight $w_G(u,v_i)-w_G(u,v_1)\ge0$. It is straightforward to verify all of the required conditions.
	\end{proof}
	
	During an iteration of our algorithm, we will compute a valid reweighting $\phi_{\tau}$, which may cause a negative edge in $G_{\tau}$ to become non-negative. However, even if a negative edge in $G_{\tau}$ is non-negative after the reweighting $\phi_{\tau}$, we will still call this edge negative through the rest of the iteration and count it towards hops in $h$-hop shortest paths. Following prior work, we refer to this as \emph{freezing} the set of negative edges before applying a reweighting. 
	
	In an iteration, we also add \emph{imaginary} negative edges to the graph. The heads of imaginary negative edges will not be considered a negative vertex, and will not be included in $N$, but will still be counted towards hops in $h$-hop shortest paths. This distinction is important when we apply betweenness reduction: hop-distance will be measured with respect to all negative edges but betweenness is only measured with respect to real ones (via the set $N$).

	\section{Technical Overview}\label{sec:technical-overview}
	
	We begin by reviewing the shortcutting approach by \cite{li2025shortcutting}. The first step is to compute a reweighting $\phi$ so that the negative paths in the graph have more structure. A similar betweenness reduction procedure has been key to all previous work beating Bellman-Ford. We remark that the notation below is slightly different from \cite{li2025shortcutting}.

	\begin{definition}
		For $s,t\in V$, the weak betweenness of $(s,t)$ is the number of negative vertices $r\in V$ such that $d^0(s,r)+d^{1}(r,t)<0$.
	\end{definition}
	\begin{lemma}[Betweenness reduction]\label[lemma]{lem:betweenness-reduction}
		Consider a graph $G$ with $k$ negative vertices, and freeze the negative edges. For any parameter $b\ge1$, there is a randomized algorithm that returns either a set of valid potentials $\phi$ or a negative cycle, such that with high probability, all pairs $(s,t)\in V\times V$ have weak betweenness at most $k/b$ under the new weights $w_\phi$. The algorithm makes one call to negative-weight single-source shortest paths on a graph with $m$ edges, $n$ vertices, and $O(b\log n)$ negative vertices, and takes $O(m\log{n})$ additional time.
	\end{lemma}
	
	In the resulting graph with low betweenness, they then compute forward and backward Dijkstra searches from each negative edge. They show that these searches can be implemented simultaneously for all negative vertices efficiently in a graph with low betweenness. 
	
	\begin{lemma}[Lemma 8 of \cite{li2025shortcutting}]\label[lemma]{lem:dijkstra-both-ways}
		There is an algorithm that, for any given negative vertex $r\in V$, computes a number $\Delta_r$ and two sets $V^{out}_r$ and $V^{in}_r$ such that
		\begin{enumerate}
			\item $d^1(r,v)\le-\Delta_r$ for all $v\in V^{out}_r$,
			\item $d^1(r,v)\ge-\Delta_r$ for all $v\not\in V^{out}_r$,
			\item $d^0(v,r)\le\Delta_r$ for all $v\in V^{in}_r$,
			\item $d^0(v,r)\ge\Delta_r$ for all $v\not\in V^{in}_r$,
			\item Either (1) or (3) is satisfied with strict inequality, and
			\item $\big||V^{out}_r|-|V^{in}_r|\big|\le 1.$
		\end{enumerate}
		The algorithm runs in time $O((|V^{out}_r|+|V^{in}_r|)^2+(|V^{out}_r|+|V^{in}_r|)\log{n})$. Moreover, the algorithm can output the values of $d^1(r,v)$ for all $v\in V^{out}_r$ and $d^0(v,r)$ for all $v\in V^{in}_r$.
	\end{lemma}
	
	\begin{lemma}[Lemma 9 of \cite{li2025shortcutting}]\label[lemma]{lem:total-size-bound}
		Under the betweenness reduction guarantee of \Cref{lem:betweenness-reduction}, we have $\sum_{r\in V}(|V_r^{out}|+|V_r^{in}|)^2\le O(kn^2/b)$ and $\sum_{r\in V}(|V_r^{out}|+|V_r^{in}|)\le O(kn/\sqrt b)$.
	\end{lemma}
	
	Choosing $b=k/\text{poly}\log{n}$, \Cref{lem:dijkstra-both-ways,lem:total-size-bound} implies that we can compute sets $V_r^{out}$ and $V_r^{in}$ for all $r\in N$ in time $\tilde{O}(n^2)$. Using these sets, the algorithm \emph{shortcuts} the graph by adding \emph{Steiner vertices} with the following edges to the graph:
	\begin{enumerate}
		\item For each negative vertex $r\in V$, create a new vertex $\tilde r$.
		\item For each $r\in V$, $v\in V_r^{out}\cup\{r\}$, and $w\in N^{out}(v)$, add the edge $(\tilde r,w)$ of weight $d^1(r,v)+\Delta_r+w(v,w)$ \emph{only if it is non-negative},\label{item:shortcut-2}
		\item For each $r\in V$, $v\in V_r^{in}\cup\{r\}$, and $u\in N^{in}(v)$, add the edge $(u,\tilde r)$ of weight $w(u,v)+d^0(v,r)-\Delta_r$ \emph{only if it is non-negative},\label{item:shortcut-3}
		\item For each negative edge $(r,r')$, $u\in V_r^{out}\cup\{r\}$, if there is a (unique) negative out-edge $(u,v)$ of $u$, add the edge $(r,v)$ of weight $d^1(r,u)+w(u,v)$, and\label{item:shortcut-4}
		\item For each negative edge $(r,r')$, $v\in V_r^{in}\cup\{r\}$, if there is a (unique) negative in-edge $(u,v)$ of $v$, add the edge $(u,r')$ of weight $w(u,v)+d^0(v,r)+w(r,r')$.\label{item:shortcut-5}
	\end{enumerate}
	Uniqueness of negative out-edges and in-edges is guaranteed by preprocessing (\Cref{lem:one-negative-outgoing-edge}), and steps (4) and (5) can be implemented in $2\sum_{r\in V}(|V_r^{out}|+|V_r^{in}|)=\tilde{O}(n^{1.5})$ time. The bottleneck of the shortcutting algorithm is steps (2) and (3), which takes up to $2\sum_{r\in V}(|V_r^{out}|+|V_r^{in}|)n=\tilde{\Theta}(n^{2.5})$ time. First, we sketch that the shortcutting procedure above reduces the number of hops in shortest paths by a constant factor. This will motivate our faster construction.
	
	\begin{figure}[t]\centering
		\includegraphics[scale=.8]{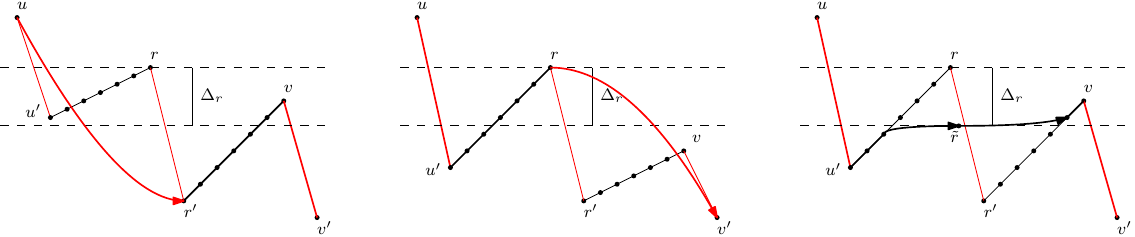}
		\caption{The three cases in the proof of \Cref{lem:shortcut}, where height indicates cumulative distance to each vertex. Negative edges are marked in red, and the new shortcut path is marked in bold.
		}\label{fig:shortcut2}
	\end{figure}
	
	\begin{lemma}\label{lem:shortcut}
		Consider any $s,t\in V$ and a shortest $(s,t)$-path $P$ with $h$ negative edges. After the shortcutting procedure above, there is an $(s,t)$-path with at most $h-\lfloor h/3\rfloor$ negative edges, and whose weight is at most that of $P$.    
	\end{lemma}
	\begin{proof}[Proof Sketch.]
		Consider the negative edges of $P$ in the order from $s$ to $t$. Create $\lfloor h/3\rfloor$ disjoint groups of three consecutive negative edges; we will shortcut each group using at most two negative edges, so that the shortcut path has at most $h-\lfloor h/3\rfloor$ negative edges. Let $(u,u')$, $(r,r')$, and $(v,v')$ be such a group of negative edges, in that order. There are three cases:
		\begin{enumerate}
			\item $d^0(u',r)<\Delta_r$. In this case, \Cref{lem:dijkstra-both-ways} ensures that $u'\in V_r^{in}$, so step~(\ref{item:shortcut-5}) adds the edge $(u,r')$ of weight $w(u,u')+d^0(u',r)+w(r,r')$, which is exactly the weight of the $u\leadsto r'$ segment. Thus, we shortcut the $u\leadsto r'$ segment of $P$ via the single (negative) edge $(u,r')$.
			\item $d^1(r,v)<-\Delta_r$. In this case, \Cref{lem:dijkstra-both-ways} ensures that $v\in V_r^{out}$, so step~(\ref{item:shortcut-4}) adds the edge $(r,v')$ of weight $d^1(r,v)+w(v,v')$, which is exactly the weight of the $r\leadsto v'$ segment. Thus, we shortcut the $r\leadsto v'$ segment of $P$ via the single (negative) edge $(r,v')$. 
			\item $d^0(u',r)\ge\Delta_r$ and $d^1(r,v)\ge-\Delta_r$. Consider the $u'\leadsto r$ segment of $P$, and let $u''$ be the last vertex on this segment with $d^0(u'',r)\ge\Delta_r$. For simplicity, assume that $u''\ne r$, since otherwise the argument is trivial. Then the vertex after $u''$ exists and is in $V_r^{in}$, so the edge $(u'',\tilde r)$ is added with weight $d^0(u'',r)-\Delta_r\ge0$ on step~(\ref{item:shortcut-3}). Similarly, consider the $r'\leadsto v$ segment of $P$, and let $v''$ be the first vertex on this segment with $d^1(r,v'')\ge-\Delta_r$. Then the vertex before $v''$ is in $V_r^{out}$, so the edge $(\tilde r,v'')$ is added with weight $d^1(r,v'')+\Delta_r\ge0$ on step~(\ref{item:shortcut-2}). The $u''\to\tilde r\to v''$ path consists of two non-negative edges and has weight $d^0(u'',r)+d^1(r,v'')$, which is exactly the weight of the $u''\leadsto r\leadsto v''$ segment of $P$. Thus, we can shortcut this segment of $P$ by the path $u''\to\tilde r\to v''$ with no negative edges.\qedhere
		\end{enumerate}
	\end{proof}
	
	Cases (1) and (2) of the proof above only use shortcut edges added in steps~(\ref{item:shortcut-4}) and (\ref{item:shortcut-5}). Since steps~(\ref{item:shortcut-4}) and (\ref{item:shortcut-5}) can already be implemented in $\tilde{O}(n^{1.5})$ time, these are the easy cases. We now focus on case (3), which requires the shortcut edges added in steps~(\ref{item:shortcut-2}) and (\ref{item:shortcut-3}).
 
	First, suppose that we add \emph{imaginary} edges $(\tilde{r},v)$ for each $v\in V^{out}_r$ with weight $\omega'(\tilde{r},v)=d^1(r,v)+\Delta_r$ and imaginary edges $(v,\tilde{r})$ with weight $\omega'(v,\tilde{r})=d^0(v,r)-\Delta_r$ for each $v\in V^{in}_r$. We use weight function $\omega'$ to distinguish imaginary edges. The key properties needed in the proof which are guaranteed by adding the shortcut edges in steps (\ref{item:shortcut-2}) and (\ref{item:shortcut-3}) are: 
	\begin{itemize}
		\item for each $r\in N$, each imaginary edge $(\tilde{r},v)$, and each out-edge $(v,w)$ of $v$, \\ if $\omega'(\tilde{r},v)+\omega(v,w)\ge 0$, then step~(\ref{item:shortcut-2}) adds a non-negative shortcut edge $(\tilde{r},w)$.
		\item for each $r\in N$, each imaginary edge $(v,\tilde{r})$, and each in-edge $(u,v)$ of $v$, \\ if $\omega(u,v)+\omega'(v,\tilde{r})\ge 0$, then step~(\ref{item:shortcut-3}) adds a non-negative shortcut edge $(u,\tilde{r})$.
	\end{itemize}
	Consider the two-edge paths $\tilde{r}\to v\to w$ and $u\to v\to\tilde{r}$ above, which consist of one imaginary edge. We call them \emph{locally-negative paths} since even though the imaginary edge is negative, its effect on the shortest paths is only local. Steps (\ref{item:shortcut-2}) and (\ref{item:shortcut-3}) achieve the property of shortcutting locally-negative paths using paths of length 1 by brute-forcing over the entire vertex neighborhood, which takes $\tilde{O}(n^{2.5})$ time. We will sketch a construction which, by adding $2n$ additional Steiner nodes to the graph, achieves these same shortcutting guarantees, but in $\tilde{O}(n^2)$ time.
	
	In this sketch, we assume $|V^{in}_r|+|V^{out}_r|\le\kappa=\tilde{O}(\sqrt{n})$ for each negative vertex $r\in N$ for a simpler presentation. Intuitively, this is guaranteed on average by applying betweenness reduction with parameter $b=k/\text{poly}\log{n}$. We will only focus on achieving the property for in-searches, since out-searches are symmetric. The idea is to add a new vertex $v_{in}$ for each (non-negative) vertex $v\in V$ to help in constructing the non-negative paths from $u$ to $\tilde{r}$. Let $\delta(v_{in})=\min_{r\in N:v\in V^{in}_r}\omega'(v,\tilde{r})$; we add the following edges to $v_{in}$ if they are non-negative:
	\begin{itemize}
		\item for every imaginary edge $(v,\tilde{r})$, add an edge $(v_{in},\tilde{r})$ with weight $\omega'(v,\tilde{r})-\delta(v_{in})$,
		\item for every non-negative edge $(u,v)$, add an edge $(u,v_{in})$ with weight $\omega(u,v)+\delta(v_{in})$.
	\end{itemize}
	Intuitively, $v_{in}$ is a \emph{copy} of $v$ which is shifted up by $\delta(v_{in})$, with resulting negative edges removed.

	We first attempt to shortcut the locally-negative path $u\to v\to \tilde{r}$ using path $u\to v_{in}\to \tilde{r}$. Since $v\in V^{in}_r$, we know $\delta(v_{in})\le \omega'(v,\tilde{r})$ by definition of $\delta(v_{in})$, implying that $\omega(v_{in},\tilde{r})\ge0$. If $\omega(u,v_{in})=\omega(u,v)+\delta(v_{in})\ge0$ as well, we can shortcut using $u\to v_{in}\to\tilde{r}$. Unfortunately, $\omega(u,v)+\delta(v_{in})$ may be negative sometimes, so just adding the copies is insufficient. The final insight is that if $\omega(u,v)+\delta(v_{in})<0$, one can show that there is some $r'\in N$ such that $u\in V^{in}_{r'}$. Using this property, we can brute force over all remaining locally-negative paths in $\tilde{O}(n^2)$ time: for each $r\in N$, each $v\in V^{in}_r$, and each $u\in V^{in}_{r'}$ for $r'\in N$ defined as above, if $\omega(u,v)+\omega'(v,\tilde{r})\ge0$, add an edge $(u,\tilde{r})$ of weight $\omega(u,v)+\omega'(v,\tilde{r})$. This shortcuts all remaining locally-negative paths, and the brute-force takes $O(k\cdot\kappa^2)\le \tilde{O}(n^2)$. 
	
	Unfortunately, since the number of vertices more than doubles every iteration, naively iterating this shortcutting procedure $O(\log{n})$ times would blow up the number of nodes in the graph to $\Omega(n^2)$. We achieve our main result by designing a similar shortcutting approach which adds $O(\log{n})$ copies of each original vertex per iteration, so the total number of vertices after iterating is $O(n\log^2{n})$. Along the way, we develop a stronger version of \Cref{lem:betweenness-reduction,lem:dijkstra-both-ways} which applies to $h$-hop paths. In particular, we believe our stronger betweenness reduction guarantees (see \Cref{sec:betweenness}) may be of independent interest for future improvements in the sparse case.
	
	\section{Shortcutting Algorithm}
	\label{sec:algorithm}
	
	As outlined in the technical overview, our algorithm is an iterative shortcutting algorithm, which proceeds over $L=\Theta(\log{k})$ iterations. In each iteration, we add some nodes and edges to the graph in order to decrease the number of negative edges on shortest paths. Let $G_{\tau}=(V_{\tau},E_{\tau})$ denote the graph after ${\tau}$ iterations of shortcutting, where $G_0=(V_0,E_0)$ is the original graph. Though this graph may have a larger vertex set $V_{\tau}\supseteq V_0$, all additional vertices are \emph{copies} of some (non-negative) vertex in $V_0$, as formalized in \Cref{sec:copy}. We prove the following:
	
	\begin{theorem}\label{thm:one-iteration-shortcutting}
		Let $G_0=(V_0,E_0)$ be the original (directed, weighted) graph and let $G_{\tau}=(V_{\tau},E_{\tau})$ be a (directed, weighted) graph such that each vertex is a copy of some vertex in $G_0$ and there are at most $c=O(\log^2{n})$ total copies of each vertex. There is a randomized algorithm that outputs a supergraph $G_{\tau+1}=(V_{\tau+1},E_{\tau+1})$ with $V_{\tau+1}\supseteq V_{\tau}$ and a potential $\phi_{\tau}$ on $V_{\tau}$ such that
		\begin{enumerate}
			\item For any $s,t\in V_{\tau}$, we have $d_{G_{\tau+1}}(s,t)=d_{G_{\tau}}(s,t)+\phi_{\tau}(s)-\phi_{\tau}(t)$,
			\item For any $s,t\in V$ and hop bound $h$, we have $d_{G_{\tau+1}}^{h-\lfloor h/3\rfloor}(s,t)\le d_{G_{\tau}}^h(s,t)+\phi_{\tau}(s)-\phi_{\tau}(t)$,
			\item $V_{\tau+1}$ has at most $O(\log{n})$ additional copies of each vertex, so $|V_{\tau+1}|\le |V_{\tau}|+O(n\log{n})$, 
			\item There are no new negative vertices.
		\end{enumerate}
		The algorithm uses one call to negative weight shortest paths on $O(|V_{\tau}|)$ nodes and $O(k/2^{\sqrt{\log{n}}})$ negative vertices, and takes $n^{2+o(1)}$ additional time.
	\end{theorem}
	
	Before we prove \Cref{thm:one-iteration-shortcutting}, we first show how to apply it to prove \Cref{thm:main}.
	
	\begin{proof}[Proof of \Cref{thm:main}]
		If $k\le 2^{\sqrt{\log{n}}}$ at any point, then naively computing $k$-hop shortest paths suffices since the running time $O(kn^2)$ meets the desired bound. Otherwise, each time we apply \Cref{thm:one-iteration-shortcutting}, the number of negative hops along shortest paths drops by a constant factor. By applying the shortcutting procedure $O(\log{k})\le O(\log{n})$ times, we guarantee that all shortest paths have at most 2 negative edges. On the resulting graph, we can run $2$-hop shortest paths to compute the desired single-source distances in the final graph. Using the reweightings $\phi_{\tau}$ and property (1), we can then recover the shortest paths in the original graph.
		
		By property (3), each iteration increases the number of vertices by $O(n\log{n})$, so all graphs have at most $O(n\log^2{n})$ vertices. Thus, the runtime is dominated by $O(\log{k})$ recursive calls to single-source shortest paths, and $n^{2+o(1)}$ additional time, giving the following recursion:
		\begin{align*}
			T(n,k)\le O(\log{n})\cdot T\left(O(n\log^2{n}),k/2^{\sqrt{\log{n}}}\right)+n^{2+o(1)}.
		\end{align*}
		There are at most $\sqrt{\log{n}}$ levels of recursion since the number of negative vertices drops by $2^{\sqrt{\log{n}}}$ in each recursive call. On recursion level $i\ge 1$, there are $O(\log^{i}n)$ recursive instances, each with at most $O(n\log^{2i}{n})$ vertices. Since $i\le \sqrt{\log{n}}$, the total size of all instances is at most $n^{1+o(1)}$ so the total runtime is still $n^{2+o(1)}$.
	\end{proof}
	
	We now give our algorithm for \Cref{thm:one-iteration-shortcutting}. The high level approach is the same as in \cite{li2025shortcutting}. First, we apply a reweighting to the graph to reduce the betweenness for every pair of vertices. In the reweighted graph, we compute forward and backward searches from each negative vertex such that every vertex in the backward search can negatively reach every vertex in the forward search using few hops. Finally, we add shortcut edges between the backward and forward searches for each negative vertex to decrease the number of hops needed in a shortest path by a constant factor. To implement this strategy more efficiently than \cite{li2025shortcutting}, we need to strengthen each of the steps, as we see next. 
	
	\subsection{Stronger Betweenness Reduction}
	
	\label{sec:betweenness}
	
	In this section, we define a stronger notion of betweenness and show how to compute a reweighting which reduces strong betweenness. Informally, the weak betweenness of $s,t\in V$ is the number of vertices $v$ such that there is an $2h$-hop negative path from $s$ to $t$ through $v$ (actually, it requires $d^h(s,v)+d^h(v,t)<0$). Strong betweenness relaxes the condition, allowing the negative path to start and end at arbitrary vertices $s',t'$, as long as the path goes through $s$ and $t$. 
	
	\begin{definition}
		For a pair $(s,t)\in V\times V$ and a subset $N\subseteq V$, its $h$-hop strong betweenness, denoted $\mathrm{SBW}_N(s,t)$, is the number of vertices $v\in N$ for which there exist $s',t'\in V$ satisfying
		\begin{align*}
			d^h(s',s)+d^h(s,v)+d^h(v,t)+d^h(t,t')<0.
		\end{align*}
	\end{definition}
	We remark that strong betweenness is defined with respect to a set of vertices $N\subseteq V$. In our paper, this will always be the set of negative vertices $N$ in the original graph, hence the notation, so the dependence $N$ is often left implicit. We emphasize that $N$ need not contain the endpoint of every negative edge in the graph. Indeed, we will be applying betweenness reduction on a graph with additional negative edges, but with the original set of negative vertices $N$. 
	
	We adapt the algorithm for \Cref{lem:betweenness-reduction} given in \cite{li2025shortcutting} for strong betweenness reduction.

	\newcommand{\tforw}{\textrm{forward}}
	\newcommand{\tback}{\textrm{backward}}

	\begin{lemma}
		Consider a graph $G=(V,E)$ with a set of distinguished vertices $N\subseteq V$, called negative vertices. For any parameters $h\ge 1$ and $b\ge 1$, there is a randomized algorithm which returns either a set of valid potentials $\phi$ or a negative cycle, such that with high probability, all pairs $(s,t)\in V\times V$ have $\text{SBW}(s,t)\le k/b$ under the new weights $\omega_{\phi}$. The algorithm makes one call to negative-weight single-source shortest path on a graph with $O(mh)$ edges, $O(nh)$ vertices, and $O(b\log{n})$ negative edges, and takes $O(mh\log{n})$ additional time.\label[lemma]{lem:strong-bw-reduction}
	\end{lemma}
	
	\begin{proof}
		We construct a directed graph $H$ containing $4h+1$ copies of $G^+$, where $G^+$ denotes the graph $G$ with all negative edges removed.  These copies are denoted by $G_{0}$, $G^{\tforw}_{{i}}$ for $i\in[2h]$, and $G^{\tback}_{{i}}$ for $i\in[2h]$ and the copies of vertex $v$ are denoted $v_0$, $v_i^{\tforw}$ for $i\in[2h]$, and $v_i^{\tback}$ for $i\in[2h]$, respectively, for each $v\in V$. 
        
		We now define the edges between copies of $G^+$. Let $M\in\mathbb{R}_{\ge0}$ be larger than the absolute value of all edge weights. For each $u\in V$, add the following ``self-edges'', each with weight $M$:
		\begin{itemize}
			\item An edge $(u_{0},u^{\tforw}_{{1}})$ and edges $(u^{\tforw}_{{i}},u^{\tforw}_{{i+1}})$ for each $i\in[2h-1]$.
			\item An edge $(u^{\tback}_{{1}},u_{0})$ and edges $(u^{\tback}_{i+1},u^{\tback}_{i})$ for each $i\in[2h-1]$
		\end{itemize}
		Next, for each negative edge $(u,v)$, add the following edges, each with weight $\omega(u,v)+M\ge0$:
		\begin{itemize}
			\item An edge $(u_{0},v^{\tforw}_{{1}})$ and edges $(u^{\tforw}_{{i}},v^{\tforw}_{{i+1}})$ for each $i\in[2h-1]$.
			\item An edge $(u^{\tback}_{{1}},v_{0})$ and edges $(u^{\tback}_{i+1},v^{\tback}_{i})$ for each $i\in[2h-1]$
		\end{itemize}
		Finally, sample a random set $S\subseteq N$ with $O(b\log n)$ vertices; for each vertex $w\in S$, add an edge $(w^{\tforw}_{{2h}},w^{\tback}_{{2h}})$ with weight $-4hM$. We compute $\phi'(v)=d(V_H,v)$ for each $v\in V_H$ via a call to negative-weight single-source shortest path. It is easy to see that $|E_H|=O(mh)$, $|V_H|=O(nh)$, and the number of negative edges is $|S|=O(b\log{n})$.

	\begin{figure}[h]\centering
		\includegraphics[scale=.8]{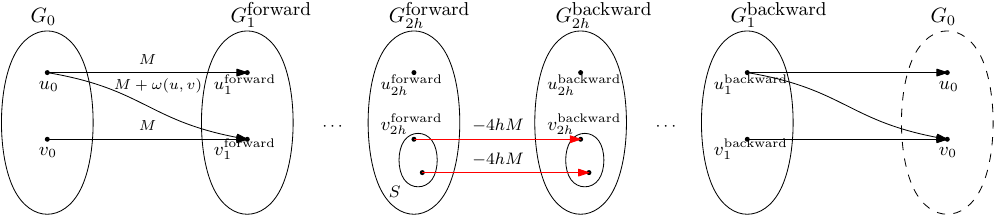}
        \label{fig:betweenness-1}
	\end{figure}

		\begin{subclaim}
			If we find a negative cycle in $H$, there is a negative cycle in $G$.
		\end{subclaim}
		\begin{subproof}
			Map each edge $(u_i,v_j)\in H$ in the negative cycle to edge $(u,v)\in G$ to obtain a closed walk in $G$. We will show that this closed walk has negative weight, so it is a negative cycle in $G$.
			First, observe that the negative cycle must traverse the copies of $G^+$ in the following order:
			\[G_{0}\to G^{\tforw}_{{1}}\to\ldots\to G^{\tforw}_{{2h}}\to G^{\tback}_{{2h}}\to\ldots\to G^{\tback}_{{1}}\to G_{0}.\] 
			For each such loop traversed, there are $4h$ edges with additional $+M$ weight and one edge, from $G^{\tforw}_{{2h}}$ to $G^{\tback}_{{2h}}$, with weight $-4hM$. Thus, the additional weights in $H$ cancel out, so the closed walk in $G$ has the same (negative) weight as the negative cycle in $H$, as desired.
		\end{subproof}
		
		If we don't find a negative cycle in $H$, then we obtain the reweighting $\phi'(v)=d_H(V,v)$ for each $v\in V_H$. Define $\phi(v)=\phi'(v_0)$ for each $v\in V$, which is a valid potential function because any non-negative edge $(u,v)$ has a copy $(u_0,v_0)$ in $G_0$, so $\phi(v)-\phi(u)=\phi'(v_0)-\phi'(u_0)\le\omega(u,v)$ by the triangle inequality on distance function $\phi'$. We now show that $\phi$ achieves the strong betweenness reduction guarantee with high probability. We start with the following subclaim.
			
		\begin{subclaim}\label[subclaim]{cl:2h-hop-is-positive}
			We have $d^{2h}_{\phi}(s',x)\ge 0$ and $d^{2h}_{\phi}(x,t')\ge 0$ for all $s',t'\in V$ and $x\in S$.
		\end{subclaim}
		\begin{subproof}
			First, we observe that for all $u,v\in V$, we have 
			\begin{align}
				\phi(v)-\phi(u)\le \min_{w\in S}d^{2h}(u,w)+d^{2h}(w,v).\label{eq:betweenness-phi}
			\end{align}
			Indeed, for any vertex $w\in S$, the path in $G$ achieving $d^{2h}(u,w)+d^{2h}(w,v)$ can be routed in $H$ by routing the $2h$-hop path $d^{2h}(u,w)$ from $u_0$ to $w^{\tforw}_{2h}$, traversing the edge $(w^{\tforw}_{{2h}},w^{\tback}_{{2h}})$, and then routing the $2h$-hop path $d^{2h}(w,v)$ from $w^{\tback}_{{2h}}$ to $v_{0}$. Thus, there is a path from $u_0$ to $v_0$ of weight $d^{2h}(u,w)+d^{2h}(w,v)$. Since $\phi$ is a distance function, any path from $u_{0}$ to $v_{0}$ has weight at least $\phi(v_{0})-\phi(u_{0})$, establishing \Cref{eq:betweenness-phi}.
			
			For any $s'\in V$ and $x\in S$, taking $u=s'$ and $v=x$ in \Cref{eq:betweenness-phi} implies that
			\begin{align*}
				\phi(x)-\phi(s')\le\min_{w\in S}d^{2h}(s',w)+d^{2h}(w,x)\le d^{2h}(s',x)+d^{2h}(x,x)\le d^{2h}(s',x).
			\end{align*}
			Similarly, for any $x\in S$ and $t'\in S$, taking $u=x$ and $v=t'$ in \Cref{eq:betweenness-phi} implies that
			\begin{align*}
				\phi(t')-\phi(x)\le\min_{w\in S}d^{2h}(x,w)+d^{2h}(w,t')\le d^{2h}(x,x)+d^{2h}(x,t')\le d^{2h}(x,t').
			\end{align*}
			Rearranging implies the desired results.  
		\end{subproof}	
		
		Using the subclaim, we establish the strong betweenness reduction guarantee. Order the vertices by value $d^h(s,v)+d^h(v,t)$ and consider the sampled (negative) vertex $x\in S$ with the minimum such value. Since $d^h_{\phi}(s,v)+d^h_{\phi}(v,t)$ has the same relative ordering, there are at most $k/b$ negative vertices $v$ with $d^h_{\phi}(s,v)+d^h_{\phi}(v,t)<d^h_{\phi}(s,x)+d^h_{\phi}(x,t)$, with high probability. 
		For some negative vertex $v\in N$, suppose $d^h_{\phi}(s,v)+d^h_{\phi}(v,t)\ge d^h_{\phi}(s,x)+d^h_{\phi}(x,t)$. Then for any $s',t'\in V$, we have $d^{2h}_{\phi}(s',x)\ge 0$ and $d^{2h}_{\phi}(x,t')\ge 0$ by \Cref{cl:2h-hop-is-positive}, so we conclude that
		\begin{align*}
			0&\le d_{\phi}^{2h}(s',x)+d_{\phi}^{2h}(x,t')\\
			&\le d^h_{\phi}(s',s)+d_{\phi}^h(s,x)+d_{\phi}^h(x,t)+d^h_{\phi}(t,t')\\
			&\le  d^h_{\phi}(s',s)+d_{\phi}^h(s,v)+d_{\phi}^h(v,t)+d^h_{\phi}(t,t').
		\end{align*}
		Thus, negative vertices $v$ counted in $\text{SBW}(s,t)$ must have $d^h_{\phi}(s,v)+d^h_{\phi}(v,t)< d^h_{\phi}(s,x)+d^h_{\phi}(x,t)$. By the argument above, there can only be at most $k/b$ such vertices, as desired.
	\end{proof}

	\subsection{Forward and Backward Searches}
	
	This section presents a multi-hop generalization of \Cref{lem:dijkstra-both-ways}.
	
	\begin{lemma}\label[lemma]{lem:forward-backward}
		Given a hop bound $h$, there is an algorithm that, for any negative vertex $r\in N$, computes a number $\Delta_r$ and two sets $\tilde{V}_r^{out}$ and $\tilde{V}_r^{in}$ such that 
		\begin{enumerate}
			\item for all $v\in \tilde{V}_r^{out}$, there is an $(h+1)$-hop path starting at $r$ and using $v$ of weight $\le -\Delta_r$,\label{item:dijsktra-both-ways-1}
			\item for all $v\not\in\tilde{V}_r^{out}$, all $(h+1)$-hop paths starting at $r$ and using $v$ have weight $\ge -\Delta_r$,\label{item:dijsktra-both-ways-2}
			\item for all $v\in\tilde{V}_r^{in}$, there is an $h$-hop path ending at $r$ and using $v$ of weight $\le \Delta_r$,\label{item:dijsktra-both-ways-3}
			\item for all $v\not\in\tilde{V}_r^{in}$, all $h$-hop paths ending at $r$ and using $v$ have weight $\ge \Delta_r$,\label{item:dijsktra-both-ways-4}
			\item either (\ref{item:dijsktra-both-ways-1}) or (\ref{item:dijsktra-both-ways-3}) is satisfied with strict inequality, and \label{item:dijsktra-both-ways-5}
			\item $\big||\tilde{V}_r^{out}-\tilde{V}_r^{in}|\big|\le 1$.\label{item:dijsktra-both-ways-6}
		\end{enumerate}
		The algorithm uses $\tilde{O}(hn^2)$ preprocessing time for fixed $h$, and on query $r\in N$, runs in time $O\big((|\tilde{V}_r^{out}|+|\tilde{V}_r^{in}|)^2\cdot h^2+(\tilde{V}_r^{out}+\tilde{V}_r^{in})\cdot h\log{n}\big)$. Moreover, the algorithm can output the values of $d^{h+1}(r,v)$ for all $v\in\tilde{V}_r^{out}$ and $d^h(v,r)$ for all $v\in\tilde{V}_r^{in}$.
	\end{lemma}
	\begin{proof}
		We will run two Dijkstra's algorithms in parallel in two separate graphs, both with source $r$. To define these graphs, we first compute $d^{i}(V,v)$ and $d^{i}(v,V)$ for each $v\in V$, $i\in\{0,\ldots,h\}$, which can be done in $\tilde{O}(hn^2)$ time. Informally, $H_{out}$ and $H_{in}$ are layered graphs with non-negative edges defined so we can use Dijkstra to simulate $h$-hop shortest paths. Below, recall that $G^+$ denotes the graph $G$ with all negative edges removed.
		\begin{enumerate}
			\item The first graph $H_{out}$ starts with $h+1$ copies of $G^+$, which we denote $G_i^{out}$ for $i\in\{0,\ldots,h\}$. Between copies $G_{i-1}^{out}$ and $G_{i}^{out}$ for each $i\in[h]$, we add edges $(u_{i-1},v_{i})$ with weight $\omega(u,v)$ for each negative edge $(u,v)$ and also ``self-edges'' $(u_{i-1},u_i)$ with weight $0$ for each vertex $u\in V$. Finally, we apply a potential function $\phi_{out}$ on $H_{out}$ to guarantee that all edge weights are non-negative. Define $\phi_{out}(v_i)=-d_G^{h-i}(v,V)$ for each $i\in\{0,\ldots,h\}$ and $v_i\in V$. We use $\phi_{out}$ to reweight the graph and denote the resulting graph as $H_{out}$.

			We now show that all edges are non-negative after the reweighting. For any non-negative edge $(u,v)\in G^+$, there are $h+1$ copies of the edge in $H_{out}$ of the form $(u_i,v_i)$ for some $i\in\{0,\ldots,h\}$. After the reweighting, this edge will have weight
			$$\omega_{H_{out}}(u_i,v_i)=\omega(u,v)-d_G^{h-i}(u,V)+d_G^{h-i}(v,V)\ge 0.$$
			For any negative edge $(u,v)$, there are $h$ copies of the edge in $H_{out}$ of the form $(u_{i-1},v_i)$ for some $i\in[h]$. After the reweighting, this edge will have weight
			\begin{align*}
				\omega_{H_{out}}(u_{i-1},v_i)=\omega(u,v)-d_G^{h+1-i}(u,V)+d_G^{h-i}(v,V)\ge 0.
			\end{align*}
			Finally, for every vertex $u$, there is a self-edge $(u_{i-1},u_i)$ for each $i\in[h]$. After the reweighting, this edge will have weight
			$$\omega_{H_{out}}(u_{i-1},u_{i})=d_G^{h-i}(u,V)-d_G^{h-i+1}(u,V)\ge 0.$$
			The non-negativity of each of these edges follows by the triangle inequality.
			\item The second graph $H_{in}$ starts with $h+1$ copies of $G^+$ with edges reversed, which we denote $G^i_{in}$ for $i\in\{0,\ldots,h\}$. Between copies $G^{i-1}_{in}$ and $G^{i}_{in}$ for each $i\in[h]$, we add edges $(v_{i-1},u_{i})$ with weight $\omega(u,v)$ for each negative edge $(u,v)$ and also self-edges $(u_{i-1},u_i)$ with weight $0$ for each vertex $u\in V$. We also define a potential function $\phi_{in}$ on $H_{in}$ to guarantee that all edge weights are non-negative, which is defined as follows. For $i\in\{0,\ldots,h\}$ and each copy $v_i\in V$, we set $\phi_{in}(v_i)=-d_G^{h-i}(V,v)$. We use $\phi_{in}$ to reweight the graph and denote the resulting graph as $H_{in}$.
			
			We now show that all edges are non-negative after the reweighting.
			For any non-negative edge $(u,v)\in G^+$, there are $h+1$ copies of the edge in $H_{in}$ of the form $(v_i,u_i)$ for some $i\in\{0,\ldots,h\}$. After the reweighting, this edge will have weight $$\omega_{H_{in}}(v_{i},u_{i})=\omega(u,v)-d_G^{h-i}(V,v)+d_G^{h-i}(V,u)\ge 0.$$
			For any negative edge $(u,v)$, there are $h$ copies of the edge in $H_{in}$ of the form $(v_{i-1},u_i)$ for some $i\in[h]$. After the reweighting, this edge will have weight
			$$\omega_{H_{in}}(v_{i-1},u_{i})=\omega(u,v)-d_G^{h-i+1}(V,v)+d_G^{h-i}(V,u)\ge 0.$$
			Finally, for every vertex $u$, there is a self loop $(u_{i-1},u_i)$ for each $i\in[h]$. After reweighting, this edge will have weight 
			$$\omega_{H_{in}}(u_{i-1},u_i)=d_G^{h-i}(V,u)-d_G^{h-i+1}(V,u)\ge 0.$$
			The non-negativity of each of these edges follows by the triangle inequality.
		\end{enumerate}
		In the pre-processing stage, we construct graphs $H_{in},H_{out}$ and compute potentials $\phi_{in},\phi_{out}$. This takes a total of $\tilde{O}(hn^2)$ time.
		
		Next, consider a negative vertex $r\in N$. By the properties of negative vertices, we know there is a single out-edge $(r,r')$ of $r$, and this edge may be negative. We add an additional edge $(r_0,r'_0)$ to $H_{out}$ with weight $\omega(r,r')+\phi_{out}(r_0)-\phi_{out}(r'_0)$. This completes the definition of the graphs $H_{in}$ and $H_{out}$ which we perform searches on. Before stating our algorithm, we characterize the shortest path distances from $r_0$ in $H_{in}$ and $H_{out}$:
		\begin{align}
			d_{H_{out}}(r_0,v_i)&=d^{i+1}_G(r,v)+\phi_{out}(r_0)-\phi_{out}(v_i)=d^{i+1}_G(r,v)-d^h_G(r,V)+d^{h-i}_G(v,V)\label{eq:dist-h-out}\\
			d_{H_{in}}(r_0,v_i)&=d^i_G(v,r)+\phi_{in}(r_0)-\phi_{in}(v_i)=d^i_G(v,r)-d^h_G(V,r)+d^{h-i}_G(V,v)\label{eq:dist-h-in}
		\end{align}
		The expressions follow from the layered graph structure of $H_{in}$ and $H_{out}$. We note that the reason shortest paths to $v_i$ in $H_{out}$ correspond to $(i+1)$-hop shortest paths in $G$ is due to the additional negative edge $(r_0,r'_0)$ we add. All paths starting from $r_0$ in $H_{out}$ use this edge, increasing the hop count by 1. 
		
		Now, we perform the Dijkstra searches on $H_{in}$ and $H_{out}$, which is valid since $H_{in}$ has non-negative edges and the only (possibly) negative edge in $H_{out}$ is an out-edge from the source. Recall that Dijkstra's algorithm maintains a subset of processed vertices, a minimum distance $d$ to a processed vertex, and has the guarantee that $d(r,v)\le d$ for all processed vertices $v$. We alternate processing vertices in the two instances of Dijkstra's algorithm. Let $\tilde{V}_r^{in}$ and $\tilde{V}_r^{out}$ to be the set of vertices that have at least one copy processed in $H_{in}$ and $H_{out}$, respectively. When we process a vertex $v_i$ in $H_{in}$ (and similarly for $H_{out}$), we add $v$ into $\tilde{V}_r^{in}$ if it is not already in this set. Each step, we process on $H_{in}$ if $\tilde{V}_r^{in}$ currently contains fewer vertices compared to $\tilde{V}_r^{out}$, and $H_{out}$ otherwise. Let $d_{in}$ and $d_{out}$ denote the minimum distance to an unprocessed vertex in $H_{in}$ and $H_{out}$, respectively. We terminate the algorithm when $d_{in}+d_{out}\ge \phi_{out}(r_0)+\phi_{in}(r_0)$, and output $\tilde{V}_r^{in}$ and $\tilde{V}_r^{out}$ as the processed vertices in $H_{in}$ and $H_{out}$, respectively. Since we are alternating processing vertices in the two instances of Dijkstra's algorithm, we immediately have property (6): $\big||\tilde{V}_r^{out}|-|\tilde{V}_r^{in}|\big|\le 1$. If the last graph processed was $H_{out}$, the algorithm sets $\Delta_r=d_{in}-\phi_{in}(r_0)$; otherwise, the last graph processed was $H_{in}$ and the algorithm sets $\Delta_r=-d_{out}+\phi_{out}(r_0)$.
		
		Suppose the last graph processed was $H_{out}$, so $\Delta_r=d_{in}-\phi_{in}(r_0)$. Let $d_{out}'<d_{out}$ be the maximum distance to a processed vertex in $H_{out}$. Since the algorithm didn't terminate before processing that vertex, we have $d_{in}+d'_{out}<\phi_{in}(r_0)+\phi_{out}(r_0)$. We now verify the properties:
	
		\begin{enumerate}
			\item For all vertices $v\in\tilde{V}_r^{out}$, some copy $v_i$ of $v$ satisfies 
			$$d_{H_{out}}(r_0,v_i)\le d'_{out}<-d_{in}+\phi_{out}(r_0)+\phi_{in}(r_0)=-\Delta_r-d^h_G(r,V),$$
			where the first inequality is guaranteed by Dijkstra and the rest are by definition. Plugging in (\ref{eq:dist-h-out}) for $d_{H_{out}}(r_0,v_i)$ and rearranging gives $d^{i+1}_G(r,v)+d^{h-i}_G(v,V)<-\Delta_r$, implying that there exists some $v'\in V$ such that an $(h+1)$-hop path from $r$ to $v'$ goes through $v$ and has total weight less than $-\Delta_r$. Note that this also satisfies property (5).
			\item For all vertices $v\not\in\tilde{V}_r^{out}$, all copies $v_i$ of $v$ satisfy
			\begin{align*}
				d_{H_{out}}(r_0,v_i)\ge d_{out}\ge -d_{in}+\phi_{out}(r_0)+\phi_{in}(r_0)=-\Delta_r-d^h_G(r,V),
			\end{align*}
			where the first inequality is guaranteed by Dijkstra and the rest are by definition. Plugging in (\ref{eq:dist-h-out}) for $d_{H_{out}}(r_0,v_i)$ and rearranging gives $d^{i+1}_G(r,v)+d_G^{h-i}(v,V)\ge-\Delta_r$ for $i\in\{0,\ldots,h\}$, implying that all $(h+1)$-hop paths starting at $r$ and using $v$ have weight at least $-\Delta_r$.
			\item For all vertices $v\in \tilde{V}_r^{in}$, some copy $v_i$ of $v$ satisfies 
			$$d_{H_{in}}(r_0,v_i)\le d_{in}=\Delta_r+\phi_{in}(r_0)=\Delta_r-d^h_G(V,r),$$
			where the first inequality is guaranteed by Dijkstra and the rest are by definition. Plugging in (\ref{eq:dist-h-in}) for $d_{H_{in}}(r_0,v_i)$ and rearranging gives $d^{i}_G(v,r)+d^{h-i}_G(V,v)\le \Delta_r$, implying that there exists some $v'\in V$ such that an $h$-hop path from $v'$ to $r$ goes through $v$ and has total weight at most $\Delta_r$. 
			\item For all vertices $v\not\in \tilde{V}_r^{in}$, all copies $v_i$ of $v$ satisfy 
			$$d_{H_{in}}(r_0,v_i)\ge d_{in}=\Delta_r+\phi_{in}(r_0)=\Delta_r-d^h_G(V,r),$$ 
			where the first inequality is guaranteed by Dijkstra and the rest are by definition. Plugging in (\ref{eq:dist-h-in}) for $d_{H_{in}}(r_0,v_i)$ and rearranging gives $d^{i}_G(v,r)+d^{h-i}_G(V,v)\ge \Delta_r$ for $i\in\{0,\ldots,h\}$, implying that all $h$-hop paths ending at $r$ and using $v$ have weight at least $\Delta_r$. 
		\end{enumerate}
		
		The case where the last graph processed was $H_{in}$ and $\Delta_r=-d_{out}+\phi_{out}(r_0)$ is similar. Let $d'_{in}<d_{in}$ be the maximum distance to a processed vertex in $H_{in}$. Since the algorithm didn't terminate before processing the vertex, we have $d'_{in}+d_{out}<\phi_{in}(r_0)+\phi_{out}(r_0)$. We now verify the properties again, with similar proofs:
		\begin{enumerate}
			\item For all vertices $v\in\tilde{V}_r^{in}$, some copy $v_i$ of $v$ satisfies
			$$d_{H_{in}}(r_0,v_i)\le d'_{in}<-d_{out}+\phi_{in}(r_0)+\phi_{out}(r_0)=\Delta_r-d^{h}_G(V,r),$$
			where the first inequality is guaranteed by Dijkstra and the rest are by definition. Plugging in (\ref{eq:dist-h-in}) for $d_{H_{in}}(r_0,v_i)$ and rearranging gives $d^i_G(v,r)+d^{h-i}_G(V,v)< \Delta_r$, implying that there exists some $v'\in V$ such that an $h$-hop path from $v'$ to $r$ goes through $v$ and has total weight less than $\Delta_r$. Note that this also satisfies property (5).
			\item For all vertices $v\not\in\tilde{V}_r^{in}$, all copies $v_i$ of $v$ satisfy 
			\begin{align*}
				d_{H_{in}}(r_0,v_i)\ge d_{in}\ge -d_{out}+\phi_{in}(r_0)+\phi_{out}(r_0)=\Delta_r-d^h_G(V,r),
			\end{align*}
			where the first inequality is guaranteed by Dijkstra and the rest are by definition. Plugging in (\ref{eq:dist-h-in}) for $d_{H_{in}}(r_0,v_i)$ and rearranging gives $d^i_G(v,r)+d^{h-i}_G(V,v)\ge \Delta_r$ for $i\in\{0,\ldots,h\}$, implying that all $h$-hop paths ending at $r$ and using $v$ have weight at least $\Delta_r$.
			\item For all vertices $v\in\tilde{V}_r^{out}$, some copy $v_i$ of $v$ satisfies
			$$d_{H_{out}}(r_0,v_i)\le d_{out}=-\Delta_r+\phi_{out}(r_0)=-\Delta_r-d^{h}_G(r,V),$$
			where the first inequality is guaranteed by Dijkstra and the rest are by definition. Plugging in (\ref{eq:dist-h-out}) for $d_{H_{out}}(r_0,v_i)$ and rearranging gives $d^{i+1}_G(r,v)+d^{h-i}_G(v,V)\le -\Delta_r$, implying that there exists some $v'\in V$ such that an $(h+1)$-hop path from $r$ to $v'$ goes through $v$ and has total weight at most $-\Delta_r$.
			\item For all vertices $v\not\in\tilde{V}_r^{out}$, all copies $v_i$ of $v$ satisfy
			$$d_{H_{out}}(r_0,v_i)\ge d_{out}=-\Delta_r+\phi_{out}(r_0)=-\Delta_r-d^{h}_G(r,V),$$
			where the first inequality is guaranteed by Dijkstra and the rest are by definition. Plugging in (\ref{eq:dist-h-out}) for $d_{H_{out}}(r_0,v_i)$ and rearranging gives $d^{i+1}_G(r,v)+d^{h-i}_G(v,V)\ge -\Delta_r$ for $i\in\{0,\ldots,h\}$, implying that all $(h+1)$-hop paths starting at $r$ and using $v$ have weight at least $-\Delta_r$. 
		\end{enumerate}
		Finally, we discuss the runtime. Naively, Dijkstra's algorithm on $H_{out}$ runs in $O(|\tilde{V}_r^{out}|\cdot hn+|\tilde{V}_r^{out}|\cdot h\log{n})$ time using a Fibonnaci heap: each of the at most $h\cdot |\tilde{V}_r^{out}|$ processed vertices inserts its outgoing neighbors into the heap, and we pop from the heap once for each processed vertex. We can speed this up to $O(|\tilde{V}_r^{out}|^2\cdot h^2+|\tilde{V}_r^{out}|\cdot h\log{n}|)$ time via the same lazy insertion trick as in \cite{li2025shortcutting}. Instead of inserting all outgoing neighbors immediately after processing a vertex, only insert the unprocessed neighbor with the smallest outgoing edge, and once that neighbor is processed, insert the next one. Here, we are assuming that each vertex has its outgoing edges initially sorted by weight. With this speedup, only $O(|\tilde{V}_r^{out}|^2\cdot h^2)$ edges are ever checked by Dijkstra's algorithm, and the running time follows. The analysis on $H_{in}$ is identical, giving the final result.
	\end{proof}
	
	\begin{lemma}\label[lemma]{lem:V-out-in-size}
		Under the $(h+1)$-hop strong betweenness reduction guarantee of \Cref{lem:strong-bw-reduction} with parameter $b\le k$, we have $\sum_{r\in V}(|\tilde{V}_r^{out}|+|\tilde{V}_r^{in}|)^2\le O(kn^2/b)$.
	\end{lemma}
	\begin{proof}
		We fix a negative vertex $r$ and consider its $\tilde{V}_r^{in}$ and $\tilde{V}_r^{out}$. For any vertex $s\in \tilde{V}_r^{in}$, \Cref{lem:forward-backward} says that we can find an $h$-hop path of weight $\le \Delta_r$ from some vertex to $r$ that uses $v$. Let the start of this path be $s'$, then we have $d^h(s',s)+d^h(s,r)\le \Delta_r$. Similarly, for every vertex $t\in \tilde{V}_r^{out}$, we can find a vertex $t'$ that $d^{h+1}(r,t)+d^{h+1}(t,t')\le -\Delta_r$. Furthermore, by property (\ref{item:dijsktra-both-ways-5}) of \Cref{lem:forward-backward}, one of these inequalities must be strict. Therefore, for every negative vertex $r$ and every $s\in \tilde{V}_r^{in}$ and $t\in \tilde{V}_r^{out}$, there exists $s',t'\in V$ such that
		$$d^h(s',s)+d^h(s,r)+d^{h+1}(r,t)+d^{h+1}(t,t')< 0,$$ 
		so $r\in\mathrm{SBW}(s,t)$. By charging each such $(s,t,r)$ tuple to $r\in\mathrm{SBW}(s,t)$, we have
		\begin{align*}
			\sum_{r\in V} |\tilde{V}_r^{in}|\cdot |\tilde{V}_r^{out}|\le\sum_s\sum_t|\mathrm{SBW}(s,t)|\le n^2\cdot(k/b).
		\end{align*}
		By property (\ref{item:dijsktra-both-ways-6}) of \Cref{lem:forward-backward}, $\left||\tilde{V}_r^{in}|-|\tilde{V}_r^{out}|\right|\le 1$, so 
		\[|\tilde{V}_r^{in}|+|\tilde{V}_r^{out}|\le O\left(\sqrt{|\tilde{V}_r^{in}|\cdot |\tilde{V}_r^{out}|+1}\right),\] 
		where $+1$ handles the case where $|\tilde{V}_r^{in}|=0$ or $|\tilde{V}_r^{out}|=0$. Therefore,
		\begin{align*}
			\sum_{r\in V} (|\tilde{V}_r^{in}|+|\tilde{V}_r^{out}|)^2&\le \sum_{r\in V} \left(\sqrt{|\tilde{V}_r^{in}|\cdot |\tilde{V}_r^{out}|+1}\right)^2\\
			&\le \sum_{r\in V} O(|\tilde{V}_r^{in}|\cdot |\tilde{V}_r^{out}|+1)\le O(kn^2/b)
		\end{align*}
		where the final inequality uses the assumption that $b\le k$. 
	\end{proof}
	
	\subsection{The Shortcutting Algorithm}
	
	With these tools in hand, we present our algorithm to shortcut $G_{\tau}$ by a constant factor. (Recall that $G_{\tau}$ is the graph after ${\tau}$ iterations of shortcutting.) The first step is to reduce the betweenness of the graph. For technical reasons, we don't apply betweenness reduction (\Cref{lem:strong-bw-reduction}) directly to $G_{\tau}$. Define $H_{\tau}$ by adding an edge between every pair of copies of the same vertex, based on how much the copies are shifted. Formally, for every pair of copies $v_i$ and $v_j$ of $v$, we add an edge
	$(v_i,v_j)$ of weight $\delta(v_j)-\delta(v_i)$. This creates more (imaginary) negative edges. Compute reweighting $\phi_{\tau}$ by applying the strong betweenness reduction (\Cref{lem:strong-bw-reduction}) on graph $H_{\tau}$ with parameters $h=2$, $b=k/2^{\sqrt{\log{n}}}$, and $N=N$ is the set of negative vertices. 
	
	Let $H$ and $\Gamma$ denote the graphs $H_{\tau}$ and $G_{\tau}$, respectively, after applying the reweighting $\phi_{\tau}$. Since $G_{\tau}$ is an edge-subgraph of $H_{\tau}$, we know that $\phi_{\tau}$ is a valid reweighting for $G_{\tau}$ as well. In the remainder of the section, we will assume distances are measured in $\Gamma$ unless otherwise specified (i.e., we let $\omega(u,v)=\omega_{\Gamma}(u,v)$ and $d^h(s,t)=d^h_{\Gamma}(s,t)$). After reweighting, $\delta(v_i)$ should be adjusted accordingly for each copy $v_i$ of each $v_0\in V_0$:
	$$\delta_{H}(v_i)=\delta_{\Gamma}(v_i)=\delta(v_i)-\phi_{\tau}(v_i)+\phi_{\tau}(v_0).$$
	It is easy to see that the new shifts satisfy the invariants for $\Gamma$ and $H$: 
    \begin{align*}
        \omega_{\Gamma}(u_i,v_j)&=\omega(u_i,v_j)+\phi_{\tau}(u_i)-\phi_{\tau}(v_j)\\
        &\ge d_{G_{\tau}}(u_0,v_0)+\delta(v_j)-\delta(u_i)+\phi_{\tau}(u_i)-\phi_{\tau}(v_j)\\
        &= d_{\Gamma}(u_0,v_0)-\phi_{\tau}(u_0)+\phi_{\tau}(v_0)+\delta(v_j)-\delta(u_i)+\phi_{\tau}(u_i)-\phi_{\tau}(v_j)\\
        &= d_{\Gamma}(u_0,v_0)+\delta_{\Gamma}(v_j)-\delta_{\Gamma}(u_i)
    \end{align*}
     For simplicity, we will continue to use $\delta(v_i)$ to denote the shifts after the reweighting.
	
	The second step is to compute the forward and backward searches for each $r\in N$ with parameter $h=1$ via \Cref{lem:forward-backward} on the graph $H$ with reduced betweenness. From this, we obtain some $\Delta_r$ for each $r\in N$, which we use to define $V^{out}_r$ and $V^{in}_r$: 
	\begin{align*}
		V^{out}_r&=\{v\in V:d^1(r,v)\le-\Delta_r\}\\
		V^{in}_r&=\{v\in V:d^0(v,r)\le\Delta_r\},
	\end{align*}
	where the inequality in the definition is strict based on which of (\ref{item:dijsktra-both-ways-1}) and (\ref{item:dijsktra-both-ways-3}) is satisfied with strict inequality. Observe that $d^1(r,v)\ge d^2_{H}(r,v)$ and $d^0(v,r)\ge d^1_{H}(v,r)$, so $V^{out}_r\subseteq \tilde{V}_r^{out}$ and $V^{in}_r\subseteq \tilde{V}_r^{in}$. In particular, \Cref{lem:V-out-in-size} also applies to $|V^{in}_r|$ and $|V^{out}_r|$. 
	
	Finally, we need to define $G_{\tau+1}$. We first define an auxiliary graph $G'_{\tau+1}$ with additional (possibly imaginary) negative edges, which is useful for defining $G_{\tau+1}$. Initialize $G'_{\tau+1}$ as $\Gamma$.
	\begin{itemize}
		\item[{\crtcrossreflabel{(S1)}[item:shortcut-step-1]}] For each negative vertex $r\in N$, create a new base vertex $\tilde{r}$, called a Steiner vertex. \item[{\crtcrossreflabel{(S2)}[item:shortcut-step-2]}] For each $r\in N$ and $v\in V^{out}_r\cup\{r'\}$, we add an edge $(\tilde{r},v)$ with weight $d^1(r,v)+\Delta_r$.
		\item[{\crtcrossreflabel{(S3)}[item:shortcut-step-3]}] For each $r\in N$ and $v\in V^{in}_r\cup\{r\}$, we add an edge $(v,\tilde{r})$ with weight $d^0(v,r)-\Delta_r$.
		\item[{\crtcrossreflabel{(S4)}[item:shortcut-step-4]}] For each negative vertex $r\in N$ and each $v\in V^{out}_r$, if $v$ is an endpoint of a negative edge $(v,v')$, add the edge $(r,v')$ to $G'_{\tau+1}$ with weight $d^1(r,v)+\omega(v,v')$.
		\item[{\crtcrossreflabel{(S5)}[item:shortcut-step-5]}] For each negative vertex $r\in N$ and each $u'\in V^{in}_r$, if $u'$ is the endpoint of a negative edge $(u,u')$, add the edge $(u,r')$ to $G_{\tau+1}'$ with weight $\omega(u,u')+d^0(u',r)+\omega(r,r')$. 
	\end{itemize}
	Note that $G'_{\tau+1}$ can be constructed in $O(n^2)$ time easily since we preprocessed using \Cref{lem:one-negative-outgoing-edge}.
	We will call a negative edge incident to a Steiner vertex in $G'_{\tau+1}$ an imaginary negative edge. Similarly to \Cref{sec:technical-overview}, we want to construct $G_{\tau+1}$ such that two-edge non-negative paths in $G'_{\tau+1}$ containing imaginary negative edges are simulated by some path in $G_{\tau+1}$ with only non-negative edges. First, we recall the definition of locally-negative paths:
	
	\begin{definition}
		We say that edges $(u,v)$ and $(v,w)$ form a ``locally-negative two-edge path'' if 
		(1) either $(u,v)$ or $(v,w)$ is an imaginary negative edge and (2) $\omega_{G'_{\tau+1}}(u,v)+\omega_{G'_{\tau+1}}(v,w)\ge 0$.
	\end{definition}
	
	For brevity, we will refer to these paths as \emph{locally-negative paths}, and leave implicit that the paths have two edges. We say that a locally-negative path from $u$ to $w$ in $G'_{\tau+1}$ is \emph{shortcut} if there is a path from $u$ to $w$ in $G_{\tau+1}$ of the same weight, but consisting of only non-negative edges. Our goal is to construct $G_{\tau+1}$ by adding edges and additional copies of vertices, so that all locally-negative paths in $G'_{\tau+1}$ are shortcut. We first show that shortcutting these paths suffices to decrease the number of hops needed for any shortest path in $G_{\tau}$ by a constant factor. The proof is very similar to that of \cite{li2025shortcutting}, which we sketched in \Cref{sec:technical-overview}.

	\begin{lemma}\label[lemma]{lem:shortcut-constant-factor}
		Consider any $s,t\in V$ and any shortest $(s,t)$-path $P$ in $\Gamma$ with $h$ negative edges. If $G_{\tau+1}$ has shortcut edges \ref{item:shortcut-step-4} and \ref{item:shortcut-step-5}, and shortcuts every locally-negative path in $G'_{\tau+1}$, then there is an $(s,t)$-path in $G_{\tau+1}$ with weight $d_{\Gamma}(s,t)$ and with at most $h-\lfloor h/3\rfloor$ negative edges.
	\end{lemma}
	\begin{proof}
		
		Consider the negative edges of $P$ in order from $s$ to $t$, and create $\lfloor h/3\rfloor$ disjoint groups of three consecutive negative edges. Let $(s,s')$, $(r,r')$, and $(t,t')$ be such a group of three consecutive negative edges. We claim there is an $s\leadsto t'$ path in $G_{\tau+1}$ with the same weight as the $s\leadsto t'$ path in $\Gamma$. Combining the shortcut paths of each group will imply the desired result. Our proof will be split into three cases based on $\Delta_r$. The first two cases will have essentially the same proof as in \cite{li2025shortcutting}; the third case is the more difficult one.

		\begin{enumerate}
			\item $d^0(s',r)<\Delta_r$. We have that $s'\in V^{in}_r$ so we add an edge $(s,r')$ with weight $\omega(s,s')+d^0(s',r)+\omega(r,r')$, which is at most the weight of the $s\leadsto r'$ segment of $P$ since the negative edges $(s,s')$ and $(r,r')$ are consecutive along $P$. Therefore, we can shortcut the $s\leadsto r'$ segment of $P$, which has two negative edges, using the single (negative) edge $(s,r')$.
			\item $d^1(r,t)<-\Delta_r$. We have $t\in V^{out}_r$ so we add an edge $(r,t')$ of weight $d^1(r,t)+\omega(t,t')$, which is at most the weight of the $r\leadsto t'$ segment of $P$ since the negative edges $(r,r')$ and $(t,t')$ are consecutive along $P$. Therefore, we can shortcut the $r\leadsto t'$ segment of $P$, which has two negative edges, using the single (negative) edge $(r,t')$.  
			\item $d^0(s',r)\ge \Delta_r$ and $d^1(r,t)\ge-\Delta_r$. Consider the $s'\leadsto r$ segment of $P$ and let $u$ be the last vertex on this segment with $d^0(u,r)\ge\Delta_r$. Then either $u=r$ or the vertex $u'$ after $u$ is in $V^{in}_r$. In the first case, there is an edge $(r,\tilde{r})$ added of weight $d^0(r,r)-\Delta_r=-\Delta_r\ge 0$, since the only way the first case occurs is when $\Delta_r\le 0$. In the second case, there is an imaginary negative edge $(u',\tilde{r})$ with weight $d^0(u',r)-\Delta_r<0$ added to $G'_{\tau+1}$, so the $u\to u'\to\tilde{r}$ path has total weight
			$$\omega(u,u')+\omega(u',\tilde{r})=\omega(u,u')+d^0(u',r)-\Delta_r= d^0(u,r)-\Delta_r\ge 0,$$
            where the second equality holds because there are no negative edges between $u$ and $r$ in the shortest $(s,t)$-path $P$.
			In particular, $u\to u'\to \tilde{r}$ is a locally-negative path in $G'_{\tau+1}$, so there is a path between $u$ and $\tilde{r}$ of non-negative edges in $G_{\tau+1}$ of the same weight. 
			
			Similarly, consider the $r'\leadsto t$ segment of $P$ and let $v'$ be the first vertex on this segment  with $d^1(r,v')\ge -\Delta_r$. Then either $v'=r'$ or the vertex $v$ before $v'$ is in $V^{out}_r$. In the first case, there is an edge $(\tilde{r},r')$ added of weight $\omega(r,r')+\Delta_r\ge 0$, since the only way the first case occurs is when $\Delta_r\ge -\omega(r,r')$. In the second case, there is an imaginary negative edge $(\tilde{r},v)$ added of weight $d^1(r,v)+\Delta_r<0$, so the $\tilde{r}\to v\to v'$ path has total weight
			$$\omega(\tilde{r},v)+\omega(v,v')=d^1(r,v)+\omega(v,v')+\Delta_r=d^1(r,v')+\Delta_r\ge 0.$$
			In particular, $\tilde{r}\to v\to v'$ forms a locally-negative path in $G'_{\tau+1}$, so there is a path between $\tilde{r}$ and $v'$ of non-negative edges in $G_{\tau+1}$ of the same weight.
			
			We wish to show that the concatenation of the paths $u\leadsto\tilde{r}$ and $\tilde{r}\leadsto v'$ in the previous paragraphs yields a path $u\leadsto \tilde{r}\leadsto v'$ of non-negative edges in $G_{\tau+1}$ of total weight equal to the $u\leadsto r\leadsto v'$ segment of $P$. The $u\leadsto\tilde{r}$ path has total weight $d^0(u,r)-\Delta_r$, and the $\tilde{r}\leadsto v'$ path has total weight $d^1(r,v')+\Delta_r$, so the concatenated path has total weight equal to the $u\leadsto r\leadsto v'$ segment of $P$.
		\end{enumerate}
		Concatenating together each of the 2-hop paths in $G_{\tau+1}$ which shortcut each of the $\lfloor h/3\rfloor$ disjoint 3-hop paths in $G_{\tau}$ proves the claim.
	\end{proof}
	
	The remainder of the section presents the last step of our shortcutting algorithm, which constructs $G_{\tau+1}$ satisfying the properties assumed in \Cref{lem:shortcut-constant-factor}. We first present a simpler construction, with a slower runtime of $n^{7/3+o(1)}$, to present the main ideas. We then use a bucketing trick to speed up the construction, giving an $n^{2+o(1)}$ time algorithm.
	
	\subsubsection{A Simpler \texorpdfstring{$n^{7/3+o(1)}$}~~Construction}
	
	This subsection will be dedicated to proving the following lemma.

	\begin{lemma}\label[lemma]{lem:shortcut-good-paths}
		There is an $n^{7/3+o(1)}$ time algorithm  which constructs a graph $G_{\tau+1}$ satisfying:
		\begin{enumerate}
			\item every locally-negative two-edge path in $G'_{\tau+1}$ is shortcut in $G_{\tau+1}$,
			\item $G_{\tau+1}$ has at most $2$ more copies of each original vertex, and
			\item $G_{\tau+1}$ has shortcut edges \ref{item:shortcut-step-4} and \ref{item:shortcut-step-5}.
		\end{enumerate} 
	\end{lemma}    
	\begin{proof}
		Initialize $G_{\tau+1}$ as the graph $G'_{\tau+1}$. We perform the following three steps to shortcut all locally-negative two-edge paths. After that, we remove the imaginary negative edges (i.e., the shortcut edges added in \ref{item:shortcut-step-2} and \ref{item:shortcut-step-3}) in $G_{\tau+1}$ to finish the construction. The shortcut edges added in \ref{item:shortcut-step-4} and \ref{item:shortcut-step-5} are kept since they are non-negative.

		\paragraph{Step 1.} Let $\kappa$ be a parameter we choose later, and let $N^*=\{r\in N:|\tilde{V}_r^{in}|+|\tilde{V}_r^{out}|\ge\kappa\}$. For each negative vertex $r\in N^*$, consider any in-edge $(v_i,\tilde{r})$ of $\tilde{r}$ and any non-negative in-edge $(u,v_i)$ of $v_i$. If $u\to v_i\to\tilde{r}$ forms a locally-negative path, we add an edge $(u,\tilde{r})$ with weight $\omega(u,v_i)+\omega(v_i,\tilde{r})$, shortcutting the locally-negative path. Symmetrically, for each vertex $r\in N^*$, consider any out-edge $(\tilde{r},v_i)$ of $\tilde{r}$ and any out-edge $(v_i,w)$ of $v_i$. If $\tilde{r}\to v_i\to w$ forms a locally-negative path, we add an edge $(\tilde{r},w)$ to $G_{\tau+1}$ with weight $\omega(\tilde{r},v_i)+\omega(v_i,w)$. 
		
		Now, we analyze the runtime of this step. For each in-edge $(v_i,\tilde{r})$, there are at most $n$ possible in-edges $(u,v_i)$ and the number of edges of the form $(v_i,\tilde{r})$ is $\sum_{r\in N^*}|\tilde{V}_r^{in}|$. To bound this summation, observe that since $|\tilde{V}_r^{in}|+|\tilde{V}_r^{out}|\ge\kappa$ for each negative vertex $r\in N^*$, \Cref{lem:V-out-in-size} with $b=k/2^{\sqrt{\log{n}}}$ implies that $|N^*|\le n^{2+o(1)}/\kappa^2$. Furthermore, we can apply Jensen's inequality and \Cref{lem:V-out-in-size} to obtain \begin{align}
			\Big(\frac{1}{|N^*|}\sum_{r\in N^*}|\tilde{V}_r^{in}|\Big)^2\le\frac{1}{|N^*|}\sum_{r\in N^*}|\tilde{V}_r^{in}|^2\le \frac{n^{2+o(1)}}{|N^*|}.
		\end{align}
		Taking square roots on both sides, multiplying both sides by $|N^*|$, and using $|N^*|\le n^{2+o(1)}/\kappa^2$, we get that $\sum_{r\in N^*}|\tilde{V}_r^{in}|\le n^{2+o(1)}/\kappa$ so the total runtime of this step is $n^{3+o(1)}/\kappa$. The proof for the out-searches from $r\in N^*$ is symmetric.

    \begin{figure}[t]\centering
        \includegraphics[scale=.8]{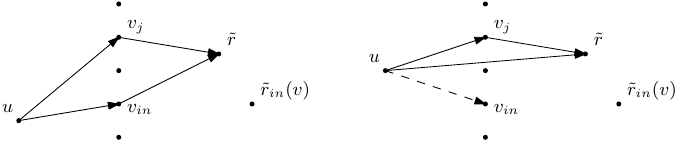}
        \caption{An illustration of steps 2 and 3 from \Cref{lem:shortcut-good-paths}, where height indicates cumulative distance, and points on the same vertical line are copies of the same vertex. The left part shows the case where the edge $(u,v_{in})$ exists for a locally-negative path $u\to v_i\to\tilde{r}$, and such paths are shortcut in step 2. The right part shows the case where this edge does not exist. In this case, $u\in \tilde V_{in}(r_{in}(v))$, and such paths are shortcut in step 3.}
        \label{fig:shortcut}
    \end{figure}
    
		\paragraph{Step 2.} Now, all locally-negative paths which aren't shortcut already must be of the form:
		\begin{itemize}
			\item $u\to v_i\to\tilde{r}$, where $r\not\in N^*$ so $|\tilde{V}_r^{in}|\le \kappa$.
			\item $\tilde{r}\to v_i\to w$, where $r\not\in N^*$ so $|\tilde{V}_r^{out}|\le\kappa$.
		\end{itemize}
		For each original vertex $v_0\in V$, we add a new copy $v_{in}$ to $G_{\tau+1}$ with shift        
		$$\delta(v_{in})=\min_{j\in[c],r\not\in N^*:v_{j}\in {V}^{in}_r}\omega(v_{j},\tilde{r})+\delta(v_{j})$$
		and a new copy $v_{out}$ to $G_{\tau+1}$ with shift
		$$\delta(v_{out})=\max_{j\in[c],r\not\in N^*:v_{j}\in V^{out}_r}-\omega(\tilde{r},v_{j})+\delta(v_{j}).$$
		We add all possible non-negative edges to and from $v_{in}$ and $v_{out}$, as defined in \Cref{cl:add-copy}. By \Cref{cl:add-copy}, we know that invariants \ref{item:graph-invariant-1}--\ref{item:graph-invariant-2} are satisfied. We further show that all edges to Steiner nodes $\tilde{r}$ are ``copied'' to $v_{in}$ and all edges from Steiner nodes $\tilde{r}$ are ``copied'' to $v_{out}$. 
		
		\begin{subclaim}\label[subclaim]{cl:fake-edges-copied}
			We have the following:
			\begin{itemize}
				\item If $(v_i,\tilde{r})\in G'_{\tau+1}$ for $r\not\in N^*$, then $\omega(v_i,\tilde{r})+\delta(v_i)-\delta(v_{in})\ge0$ so $(v_{in},\tilde{r})\in G_{\tau+1}$.
				\item If $(\tilde{r},v_i)\in G'_{\tau+1}$ for $r\not\in N^*$, then $\omega(\tilde{r},v_i)+\delta(v_{out})-\delta(v_i)\ge 0$, so $(\tilde{r},v_{out})\in G_{\tau+1}$.
			\end{itemize}
		\end{subclaim}
		\begin{subproof}
			Writing out the definition, we have
			\begin{align*}
				\omega(v_i,\tilde{r})+\delta(v_i)-\delta(v_{in})&=\max_{j\in[c],r'\not\in N^*:v_{j}\in {V}^{in}_{r'}}\omega(v_i,\tilde{r})-\omega(v_j,\tilde{r}')+\delta(v_i)-\delta(v_{j}).
			\end{align*}
			In the maximization, taking $j=i$ and $r'=r$ satisfies the constraints and gives value $0$, so the optimal choice of $j$ and $r'$ must yield a non-negative value. Thus, $C(v,\tilde{r})$ is non-empty in \Cref{cl:add-copy}, so edge $(v_{in},\tilde{r})$ is added to $G_{\tau+1}$. The proof of the second bullet is symmetric.       
		\end{subproof}
		
		We analyze what locally-negative paths are shortcut by Step 2.
		Consider a locally-negative path $u\to v_i\to \tilde{r}$, with $r\not\in N^*$. By \Cref{cl:fake-edges-copied}, we add the (non-negative) edge $(v_{in},\tilde{r})$ to $G_{\tau+1}$. Thus, if edge $(u,v_{in})$ is also added, we successfully shortcut $u\to v_i\to\tilde{r}$ using $u\to v_{in}\to\tilde{r}$. Symmetrically, consider a locally-negative path $\tilde{r}\to v_i\to w$. By \Cref{cl:fake-edges-copied}, we add the (non-negative) edge $(\tilde{r},v_{out})$ to $G_{\tau+1}$. Thus, if $(v_{out},w)$ is also added, we successfully shortcut $\tilde{r}\to v_i\to w$ using $\tilde{r}\to v_{out}\to w$. Thus, the remaining locally-negative paths which aren't already shortcut must be of the form:
		\begin{itemize}
			\item $u\to v_i\to \tilde{r}$, where $(u,v_{in})\not\in G_{\tau+1}$ and $r\not\in N^*$ so $|\tilde{V}_r^{in}|\le\kappa$.
			\item $\tilde{r}\to v_i\to w$, where $(v_{out},w)\not\in G_{\tau+1}$ and $r\not\in N^*$ so $|\tilde{V}_r^{out}|\le\kappa$.
		\end{itemize}
		We show that after adding the two copies, the remaining locally-negative paths are extremely structured, so we can brute force over them easily.
		\begin{subclaim}\label[subclaim]{cl:brute-force}
			For any such path $u\to v_i\to\tilde{r}$, we have $u\in\tilde{V}_{r'}^{in}$ for some $r'\not\in N^*$. For any such path $\tilde{r}\to v_i\to w$, we have $w\in\tilde{V}_{r'}^{out}$ for some $r'\not\in N^*$.
		\end{subclaim}
		\begin{subproof}
			First, consider a path $u\to v_i\to\tilde{r}$ and consider the negative vertex $r'$ such that
			\begin{align}\label{eq:r-prime-in}
				(j,r')=\arg\min_{j\in[c],r\not\in N^*:v_{j}\in V^{in}_r}\omega(v_{j},\tilde{r})+\delta(v_{j}).
			\end{align}
			Since $v_j\in V^{in}_{r'}$, there is a 0-hop path $v_j\leadsto r'$ in $\Gamma$ of weight $\omega(v_{j},\tilde{r}')+\Delta_{r'}$. Extend this path in $H$ via edges $(u,v_i)$ and $(v_i,v_j)$ to the path $u\to v_i\to v_j\leadsto r'$, which we denote by $P$. Since $\omega(u,v_i)\ge 0$ by our construction of locally-negative path $u\to v_i\to\tilde{r}$, this path has at most 1 negative edge ($v_i\to v_j$). Furthermore, the weight is
			$$\omega(P)= \omega(u,v_i)+\delta(v_j)-\delta(v_i)+\omega(v_j,\tilde{r}')+\Delta_{r'}.$$
			We wish to show that $\omega(P)<\Delta_{r'}$, or equivalently that $\omega(u,v_i)+\delta(v_j)-\delta(v_i)+\omega(v_j,\tilde{r}')<0$.
			
			Since edge $(u,v_{in})\not\in G_{\tau+1}$ was not added when using \Cref{cl:add-copy}, we know $C(u,v)=\emptyset$ which in particular implies that $\omega(u,v_i)+\delta(v_{in})-\delta(v_i)<0$. We also know by definition of $(j,r')$ that
			$\omega(v_j,\tilde{r}')+\delta(v_j)-\delta(v_{in})=0.$ Adding these two and canceling out $\delta(v_{in})$ gives
			$$\omega(u,v_i)-\delta(v_i)+\omega(v_j,\tilde{r}')+\delta(v_j)<0,$$
			which implies that $\omega(P)<\Delta_{r'}$, so there is a $1$-hop path from $u$ to $r'$ of weight less than $\Delta_{r'}$ and $u\in\tilde{V}_{r'}^{in}$. By a symmetric argument, we have that any locally-negative path $\tilde{r}\to v_i\to w$ which isn't shortcut already has the property that $w\in\tilde{V}_{r'}^{out}$, where 
			\begin{align}\label{eq:r-prime-out}
				(j,r')=\arg\max_{j\in[c],r\not\in N^*:v_{j}\in V^{out}_r}-\omega(\tilde{r},v_{j})+\delta(v_{j}).
			\end{align}
			We remark that the desired $r'$ for a given $v_i$ can be computed in $\tilde{O}(1)$ time in both cases.
		\end{subproof}

		\paragraph{Step 3.} To shortcut the remaining locally-negative paths, we brute force. For each $r\not\in N^*$ and each $v_i\in V^{in}_r$, consider each $u\in\tilde{V}_{r'}^{in}$ where $r'$ is as defined in (\ref{eq:r-prime-in}). If edges $(u,v_i)$ and $(v_i,\tilde{r})$ both exist and $\omega(u,v_i)+\omega(v_i,\tilde{r})\ge 0$, add an edge $(u,\tilde{r})$ to $G_{\tau+1}$ of weight $\omega(u,v_i)+\omega(v_i,\tilde{r})$. By \Cref{cl:brute-force}, this process successfully shortcuts all remaining locally-negative paths $u\to v_i\to \tilde{r}$. Symmetrically, for each $r\not\in N^*$ and each $v_i\in V^{out}_r$, we can consider each $u\in\tilde{V}_{r'}^{out}$ where $r'$ is as defined in (\ref{eq:r-prime-out}). If edges $(\tilde{r},v_i)$ and $(v_i,w)$ both exist and $\omega(\tilde{r},v_i)+\omega(v_i,w)\ge 0$, add an edge $(\tilde{r},w)$ to $G_{\tau+1}$ of weight $\omega(\tilde{r},v_i)+\omega(v_i,w)$. Again by \Cref{cl:brute-force}, this process successfully shortcuts all remaining locally-negative paths $\tilde{r}\to v_i\to w$. Since $|V^{in}_r|,|\tilde{V}_{r'}^{in}|\le\kappa$ for $r,r'\not\in N^*$ and $|V^{out}_r|,|\tilde{V}_{r'}^{out}|\le\kappa$ for $r,r'\not\in N^*$, this final shortcutting step takes $\tilde{O}(n\cdot\kappa^2)$ time. Choosing $\kappa=n^{2/3}$ balances the two terms, giving a final runtime of $n^{7/3+o(1)}$.
	\end{proof}

	\subsubsection{The \texorpdfstring{$n^{2+o(1)}$}~~Algorithm}
	
	In this subsection, we refine the ideas in \Cref{lem:shortcut-good-paths} to give an $n^{2+o(1)}$ algorithm. The remainder of the subsection will be dedicated to proving the following lemma.
	
	\begin{lemma}\label[lemma]{lem:faster-shortcut-good-paths}
		There is an $n^{2+o(1)}$ time algorithm which constructs a graph $G_{\tau+1}$ satisfying:
		\begin{enumerate}
			\item shortest path distances are preserved: $d_{G_{\tau+1}}(s,t)=d_{\Gamma}(s,t),$
			\item every locally-negative two-edge path in $G'_{\tau+1}$ is shortcut in $G_{\tau+1}$,
			\item $G_{\tau+1}$ has at most $O(\log{n})$ more copies of each original vertex,
			\item $G_{\tau+1}$ has shortcut edges \ref{item:shortcut-step-4} and \ref{item:shortcut-step-5}.
		\end{enumerate}
	\end{lemma}
	
	First, we give some intuition. For negative vertices $r\in N$ such that $|\tilde V^{in}_r|+|\tilde V^{out}_r|\le\kappa$, \Cref{lem:shortcut-good-paths} shows that we can add two copies of each original vertex and $O((|\tilde V^{in}_r|+|\tilde V^{out}_r|)\cdot\kappa)$ edges to shortcut all locally-negative paths from these negative vertices. For the remaining  negative vertices $r$, we have $|\tilde V^{in}_r|+|\tilde V^{out}_r|>\kappa$, so we use the trivial bound of $|\tilde{V}_r^{in}|+|\tilde{V}_r^{out}|\le n$. For such $r\in N$, we add $O(n(|\tilde V^{in}_r|+|\tilde V^{out}_r|))$ edges, and we show that the total number of edges added for these vertices is not too large since $\sum_r (|\tilde V^{in}_r|+|\tilde V^{out}_r|)^2=O(kn^2/b)$ by \Cref{lem:V-out-in-size}. Therefore, the running time for both parts is bounded.
	
	Our improved construction partitions the negative vertices using $O(\log{n})$ different scales $\kappa_1,\ldots,\kappa_{O(\log{n})}$, where $\kappa_1=O(|G'_{\tau+1}|)$ and $\kappa_i/\kappa_{i+1}>1$ is a constant. For each $\kappa_i$, we define
	\[N_i=\{r\in N: |\tilde V^{in}_r|+|\tilde V^{out}_r|\ge \kappa_i\}.\]
	Note that $N_1$ is empty and $N_{O(\log{n})}$ contains all the negative vertices. In the $i^{th}$ round, we assume inductively that all locally-negative paths to and from negative vertices $r\in N_i$ have already been shortcut. Our goal is to shortcut all locally-negative paths to and from negative vertices $r\in N_{i+1}\setminus N_i$, so we can proceed to the next round. Using the construction in \Cref{lem:shortcut-good-paths}, we can shortcut all locally-negative paths to and from $r\in N_{i+1}\setminus N_i$ by adding two copies of each vertex and $O(\kappa_i(|\tilde V^{in}_r|+|\tilde V^{out}_r|))$ edges. Since $r\in N_{i+1}$ and $\kappa_{i+1}/\kappa_i=\Theta(1)$, we can bound $\kappa_i\le O(|\tilde{V}_r^{in}|+|\tilde{V}_r^{out}|)$, so the number of edges added (and runtime) is $O((|\tilde V^{in}_r|+|\tilde V^{out}_r|)^2)$ for each negative vertex. By \Cref{lem:V-out-in-size} and our choice of $b$, we can bound the total number of edges added by $n^{2+o(1)}$.
	
	We formalize the algorithm sketched above. First, initialize $G_{\tau+1}$ as $G'_{\tau+1}$. For each $i\in\mathbb{N}$, define $
	\kappa_i=\lfloor(2|G_{\tau+1}'|+1)/2^{i-1}\rfloor$ and $N_i=\{r\in N: |\tilde V^{in}_r|+|\tilde V^{out}_r|\ge \kappa_i\}$.
	For each $i=1,2,\ldots,L=\lfloor\log (2|G_{\tau+1}'|+1)\rfloor+1$ in order, do the following:
	\begin{enumerate}
		\item For each original vertex $v_0\in V$, we add a new copy, denoted $v_{in}$, into $G_{\tau+1}$, with shift
		\begin{align}
			\delta(v_{in})=\min_{j\in[c],r\in N_{i+1}\setminus N_i:v_{j}\in {V}^{in}_r}\omega(v_{j},\tilde{r})+\delta(v_{j}).\label{eq:def-delta-v-in}
		\end{align}
		We add all possible edges to and from $v_{in}$, as defined by \Cref{cl:add-copy}.
		
		\item For each original vertex $v_0\in V$, we add a new copy, denoted $v_{out}$, into $G_{\tau+1}$, with shift
		\begin{align}
			\delta(v_{out})=\max_{j\in[c],r\in N_{i+1}\setminus N_i:v_{j}\in V^{out}_r}-\omega(\tilde{r},v_{j})+\delta(v_{j}).\label{eq:def-delta-v-out}
		\end{align}
		We add all possible edges to and from $v_{in}$, as defined by \Cref{cl:add-copy}.
		
		\item For each $v_0\in V$, find $(j^{in}_v,r^{in}_v)$ such that $\delta(v_{in})=\omega(v_{j^{in}_v},\tilde{r}^{in}_v)+\delta(v_{j^{in}_v})$:
		\begin{align*}
			(j^{in}_v,r^{in}_v)=\arg\min_{j\in[c],r\in N_{i+1}\setminus N_i:v_{j}\in V^{in}_r}\omega(v_{j},\tilde r)+\delta(v_{j}).
		\end{align*}
		Then, for each $r\in N_{i+1}\setminus N_i$, and each $v_j\in V^{in}_r$, we enumerate every $u\in \tilde V^{in}_{r^{in}_v}$. If the edges $(u,v_j)$ and $(v_j,\tilde r)$ both exist and $\omega(u,v_j)+\omega(v_j,\tilde r)\ge 0$, add an edge $(u,\tilde r)$ to $G_{\tau+1}$ of weight $\omega(u,v_j)+\omega(v_j,\tilde r)$.
		
		\item For each $v_0\in V$, find $(j^{out}_v,r^{out}_v)$ such that $\delta(v_{out})=-\omega(\tilde{r}^{out}_v,v_{j^{out}_v})+\delta(v_{j^{out}_v})$:
		\begin{align*}
			(j^{out}_v,r^{out}_v)=\arg\max_{j\in[c],r\in N_{i+1}\setminus N_i:v_{j}\in V^{out}_r}-\omega(\tilde r, v_{j})+\delta(v_{j}).
		\end{align*}
		Then, for each $r\in N_{i+1}\setminus N_i$, and each $v_j\in V^{out}_r$, we enumerate every $w\in \tilde V^{out}_{r^{out}_v}$. If the edges $(\tilde r, v_j)$ and $(v_j, w)$ both exist and $\omega(\tilde r, v_j)+\omega(v_j, w)\ge 0$, add an edge $(\tilde r, w)$ to $G_{\tau+1}$ of weight $\omega(\tilde r, v_j)+\omega(v_j, w)$.
	\end{enumerate}
	Finally, we remove all imaginary negative edges from $G_{\tau+1}$. By construction, $G_{\tau+1}$ adds $O(\log{n})$ copies of each vertex and has shortcut edges \ref{item:shortcut-step-4} and \ref{item:shortcut-step-5}. It remains to show that all locally-negative paths are shortcut and bound the runtime, which we prove in the following subclaims.
	
	\begin{subclaim}
		Constructing $G_{\tau+1}$ runs in time $n^{2+o(1)}$.
	\end{subclaim}
	\begin{subproof}
		We consider each round $i$ separately. In steps (1) and (2), we add copies $v_{in}$ and $v_{out}$. Computing $\delta(v_{in})$ and $\delta(v_{out})$ takes at most $O(nck)=\tilde O(n^2)$ time and adding these copies takes $O(n^2c^2)=\tilde O(n^2)$ time. In steps (3) and (4), we shortcut the remaining locally-negative paths by enumerating. We enumerate over $r\in N_{i+1}\setminus N_i$ and $v_j\in V^{in}_r\cup V^{out}_r$, and then add at most $|\tilde V^{in}_{r'}|+|\tilde{V}_{r'}^{out}|$ edges, where $r'\in N_{i+1}\setminus N_i$. Since $r'\not\in N_i$, we have $|\tilde V^{in}_{r'}|+|\tilde{V}_{r'}^{out}| \le \kappa_i$, so the running time for these steps is $O(\kappa_i(|V^{in}_r|+|V^{out}_r|))$ for every negative vertex $r\in N_{i+1}\setminus N_i$. 
		Summing over $r\in N_{i+1}\setminus N_i$, the enumeration takes time
		\begin{align*}
			\sum_{r\in N_{i+1}\setminus N_i}O(\kappa_i(|V^{in}_r|+|V^{out}_r)|)
			&\le \sum_{r\in N_{i+1}\setminus N_i}O((|\tilde V^{in}_r|+|\tilde V^{out}_r|)^2)
		\end{align*}
		where we are using (1) $r\in N_{i+1}$, so $|\tilde V^{in}_r|+|\tilde V^{out}_r|>\kappa_{i+1}=\Theta(\kappa_i)$ and (2) $|V^{in}_r|\le|\tilde{V}_r^{in}|$ and $|V^{out}_r|\le|\tilde{V}_r^{out}|$.
		Summing over all iterations, the total runtime of steps (3) and (4) is $$O\left(\sum_{r\in N} (|\tilde V^{in}_r|+|\tilde V^{out}_r|)^2\right)\le O(kn^2/b)\le n^{2+o(1)}$$ by \Cref{lem:V-out-in-size} and the choice of $b=k/2^{\sqrt{\log n}}$. Combining the time spent on adding copies and enumerations gives the desired running time.
	\end{subproof}
	
	\begin{subclaim}
		We have $d_{G_{\tau+1}}(s,t)=d_{\Gamma}(s,t)$ for each $s,t\in V$.
	\end{subclaim}
	\begin{subproof}
		Since $G_{\tau+1}$ contains $\Gamma$, it suffices to show $d_{G_{\tau+1}}(s,t)\ge d_{\Gamma}(s,t)$ for each $s,t\in V_0$.
		
		We first show that $d_{G'_{\tau+1}}(s,t)\ge d_{\Gamma}(s,t)$. Edges added in (S4) and (S5) clearly correspond to a path in $\Gamma$, so they do not decrease the shortest path distances. Next, consider paths introduced by edges added in (S2) and (S3). For any $r\in V$, note that each edge $(u,\tilde r)$ has weight at least $d(u,r)-\Delta_r$, and each added edge $(\tilde r,v)$ has weight at least $d(r,v)+\Delta_r$. Thus, any $u\to \tilde r\to v$ path through $\tilde r$ has weight at least $d(u,r)+d(r,v)\ge d(u,v)$, as promised. 
		
		Next, we show $d_{G_{\tau+1}}(s,t)\ge d_{G'_{\tau+1}}(s,t)$. In constructing $G_{\tau+1}$ from $G'_{\tau+1}$, we first iteratively add copies $v_{in}$ and $v_{out}$ via \Cref{cl:add-copy}. In particular, \Cref{cl:add-copy} implies that after all iterations, the resulting graph satisfies invariant \ref{item:graph-invariant-2}, so these copies do not decrease the shortest path distances. Finally, we remove all imaginary negative edges, which also doesn't decrease the distances, completing the proof.
	\end{subproof}
	
	\begin{subclaim}\label[subclaim]{cl:shortcut-i-i+1}
		Using the copies and edges added in the $i$-th round, every locally-negative path that uses a vertex $\tilde r$ with $r\in N_{i+1}\setminus N_i$ is shortcut.
	\end{subclaim}
	\begin{subproof}
		We first consider a locally-negative path $u\to v_j\to \tilde r$, where $r\in N_{i+1}\setminus N_i$ and $v_j\in {V}_{in}(r)$. To shortcut this path, we will use the path $u\to v_{in}\to \tilde r$ if edges $(u,v_{in})$ and $(v_{in},\tilde{r})$ both exist. Since we add edges to and from $v_{in}$ using \Cref{cl:add-copy}, we know
		\begin{itemize}
			\item we add edge $(u,v_{in})$ with weight $\omega(u,v_{in})$ if $\omega(u,v_i)+\delta(v_{in})-\delta(v_i)\ge0$ for some $i$.
			\item we add edge $(v_{in},\tilde{r})$ with weight $\omega(v_{in},\tilde{r})$ if $\omega(v_i,\tilde r)+\delta(v_i)-\delta(v_{in})\ge 0$ for some $i$.
		\end{itemize}
		We claim that the edge $(v_{in},\tilde r)$ is always added, regardless of the choice of $j,r$. Since $r\in N_{i+1}\setminus N_i$ and $v_j\in {V}_{in}(r)$, we know $\omega(v_{j},\tilde{r})+\delta(v_{j})$ is one term of the minimization when defining $\delta(v_{in})$ in \eqref{eq:def-delta-v-in}. Therefore, we have $\delta(v_{in})\le \omega(v_{j},\tilde{r})+\delta(v_{j})$, which implies $\omega(v_{j},\tilde{r})+\delta(v_{j})-\delta(v_{in})\ge 0$, proving the claim. Thus, as long as $\omega(u,v_i)+\delta(v_{in})-\delta(v_i)\ge 0$ for some $i$ as well, both edges are added and the locally-negative path is shortcut by $u\to v_{in}\to\tilde{r}$.

		Otherwise, suppose that $\omega(u,v_i)+\delta(v_{in})-\delta(v_i)<0$ for all copies $i$, so in particular we have $\omega(u,v_j)+\delta(v_{in})-\delta(v_j)<0$. Recall that in Step (3), we find negative vertex $r^{in}$ and index $j^{in}\in[c]$ such that $\delta(v_{in})=\omega(v_{j^{in}},\tilde r^{in})+\delta(v_{j^{in}})$; we leave the dependence on $v$ implicit for simplicity. Plugging this into the expression before, we get
		$$\omega(u,v_j)-\delta(v_j)+\delta(v_{j^{in}})+\omega(v_{j^{in}},\tilde r^{in})<0,$$
		which implies that the path $u\to v_j\to v_{j^{in}}\leadsto r^{in}$ in $H$ has weight less than $\Delta_{r^{in}}$ (recall that $\omega(v_{j^{in}},\tilde r^{in})=d^0(v_{j^{in}},r^{in})-\Delta_{r^{in}}$). Since this path has at most one hop (which is the edge $v_j\to v_{j^{in}}$), we have $u\in \tilde V^{in}_{r^{in}}$ by \Cref{lem:forward-backward}. Therefore, we must have found and shortcut this locally-negative path during our enumeration in Step (3) of the algorithm.
		
		The case for locally-negative paths of the form $\tilde r\to v_j\to w$, where $r\in N_{i+1}\setminus N_i$ and $v_j\in V^{out}_r$ is symmetric. To shortcut this path, we will use the path $\tilde r\to v_{out}\to w$ if edges $(\tilde{r},v_{out})$ and $(v_{out},w)$ both exist. Since we added edges to and from $v_{out}$ using \Cref{cl:add-copy}, we know
		\begin{itemize}
			\item we add edge $(\tilde{r},v_{out})$ with weight $\omega(\tilde{r},v_{out})$ if $\omega(\tilde{r},v_i)+\delta(v_{out})-\delta(v_i)\ge0$ for some $i$.
			\item we add edge $(v_{out},w)$ with weight $\omega(v_{out},w)$ if $\omega(v_i,w)+\delta(v_i)-\delta(v_{out})\ge0$ for some $i$.
		\end{itemize}
		We claim that edge $(\tilde r,v_{out})$ is always added, regardless of the choice of $j,r$. Since $r\in N_{i+1}\setminus N_i$ and $v_j\in {V}_{out}(r)$, we know $-\omega(\tilde{r},v_{j})+\delta(v_{j})$ is one term of the maximization when defining $\delta(v_{out})$ in \eqref{eq:def-delta-v-out}. Therefore, we have $\delta(v_{out})\ge -\omega(\tilde{r},v_{j})+\delta(v_{j})$, which implies $\omega(\tilde{r},v_j)+\delta(v_{out})-\delta(v_j)\ge 0$, proving the claim. Thus, as long as $\omega(v_i,w)+\delta(v_i)-\delta(v_{out})\ge0$ for some $i$ as well, both edges are added and the locally-negative path is shortcut by $\tilde{r}\to v_{out}\to w$.
		
		Otherwise, suppose that $\omega(v_i,w)+\delta(v_i)-\delta(v_{out})<0$ for all copies $i$, so in particular we have $\omega(v_j,w)+\delta(v_j)-\delta(v_{out})<0$. Recall that we find negative vertex $r^{out}$ and index $j^{out}\in[c]$ such that $\delta(v_{out})=-\omega(\tilde{r}^{out},v_{j^{out}})+\delta(v_{j^{out}})$. Plugging this into the expression before, we get 
		$$\omega(v_j,w)+\delta(v_j)+\omega(\tilde r^{out}, v_{j^{out}})-\delta(v_{j^{out}})<0,$$
		which implies that the path $r^{out}\leadsto v_{j^{out}}\to v_j\to w$ has weight at most $-\Delta_{r^{out}}$ (recall that $\omega(\tilde r^{out}, v_{j^{out}})=d^1(r^{out},v_{j^{out}})+\Delta_{r^{out}}$). Since this path has at most two hops, the edge $v_{j^{out}}\to v_j$ and one in $r^{out}\leadsto v_{j^{out}}$, so $w\in \tilde V^{out}_{r^{out}}$ by \Cref{lem:forward-backward}. Therefore, we must have found and shortcut this locally-negative path during our enumeration in Step (4) of the algorithm.
	\end{subproof}
	
	We now complete the proof. \Cref{cl:shortcut-i-i+1} proves that all locally-negative paths to or from $r\in N_{i+1}\setminus N_i$ are shortcut during the $i$-th round, for each $i\in[L]$. Since $\kappa_1=2|G'_{\tau+1}|+1$ and $\kappa_L=0$, we know $N_1=\emptyset$ and $N_L=N$. Therefore, each negative vertex $r\in N$ is in $N_{i+1}\setminus N_i$ for some $i\in[L]$, so all locally-negative paths are shortcut.

	\subsection{Completing the Shortcutting}
	Finally, we complete the proof of \Cref{thm:one-iteration-shortcutting}.
	In \Cref{sec:algorithm}, we constructed $G'_{\tau+1}$ using one recursive call to negative-weight single-source shortest paths on a graph with $O(n)$ vertices and $k/2^{\sqrt{\log{n}}}$ negative vertices, and using $n^{2+o(1)}$ additional time. 
	From $G'_{\tau+1}$, \Cref{lem:faster-shortcut-good-paths} constructs $G_{\tau+1}$ in $n^{2+o(1)}$ time satisfying the following:
	\begin{enumerate}
		\item shortest path distances are preserved: $d_{G_{\tau+1}}(s,t)=d_{\Gamma}(s,t)$,
		\item every locally-negative two-edge path in $G'_{\tau+1}$ is shortcut in $G_{\tau+1}$,
		\item $G_{\tau+1}$ has at most $O(\log{n})$ more copies of each original vertex,
		\item $G_{\tau+1}$ has shortcut edges \ref{item:shortcut-step-4} and \ref{item:shortcut-step-5}.
	\end{enumerate}
	Finally, we apply \Cref{lem:one-negative-outgoing-edge} to guarantee that each vertex has at most one negative edge incident to it. Properties (1) and (3) above immediately implies property (1) and (3) of \Cref{thm:one-iteration-shortcutting}.
	Since $d_{\Gamma}(s,t)=d_{G_{\tau}}(s,t)+\phi_{\tau}(s)-\phi_{\tau}(t)$ and similarly $d_{\Gamma}^h(s,t)=d_{G_{\tau}}^h(s,t)+\phi_{\tau}(s)-\phi_{\tau}(t)$, property (1) above implies property (1) in \Cref{thm:one-iteration-shortcutting}.
	Since $G_{\tau+1}$ satisfies properties (2) and (4) above, \Cref{lem:shortcut-constant-factor} implies that any path $P$ in $\Gamma$ with $h$-hops has a corresponding path in $G_{\tau+1}$ with $(h-\lfloor h/3\rfloor)$-hops. In particular, this implies property (4) in \Cref{thm:one-iteration-shortcutting}: $d^{h-\lfloor h/3\rfloor}_{G_{\tau+1}}(s,t)\le d^h_{\Gamma}(s,t)$.

    \section*{Acknowledgement}
    We thank Satish Rao for helpful comments and discussions. George Li is supported by the National Science Foundation Graduate Research Fellowship Program under Grant No. DGE2140739.
    
	\bibliographystyle{alpha}
	\bibliography{ref}

@inproceedings{fineman2024single,
  title={{Single-Source Shortest Paths with Negative Real Weights in $\tilde{O}(mn^{8/9})$ Time}},
  author={Fineman, Jeremy T},
  booktitle={Proceedings of the 56th Annual ACM Symposium on Theory of Computing},
  pages={3--14},
  year={2024}
}

@inproceedings{huang2025faster,
  title={{Faster Single-Source Shortest Paths with Negative Real Weights via Proper Hop Distance}},
  author={Huang, Yufan and Jin, Peter and Quanrud, Kent},
  booktitle={Proceedings of the 2025 Annual ACM-SIAM Symposium on Discrete Algorithms (SODA)},
  pages={5239--5244},
  year={2025},
  organization={SIAM}
}

@inproceedings{huang2026faster,
  title={{Faster Negative Length Shortest Paths by Bootstrapping Hop Reducers}},
  author={Huang, Yufan and Jin, Peter and Quanrud, Kent},
  booktitle={Proceedings of the 2026 Annual ACM-SIAM Symposium on Discrete Algorithms (SODA)},
  year={2026},
  organization={SIAM}
}

@article{dinitz2017hybrid,
  title={{Hybrid Bellman--Ford--Dijkstra Algorithm}},
  author={Dinitz, Yefim and Itzhak, Rotem},
  journal={Journal of Discrete Algorithms},
  volume={42},
  pages={35--44},
  year={2017},
  publisher={Elsevier}
}

@inproceedings{duan2025breaking,
  title={{Breaking the Sorting Barrier for Directed Single-Source Shortest Paths}},
  author={Duan, Ran and Mao, Jiayi and Mao, Xiao and Shu, Xinkai and Yin, Longhui},
  booktitle={Proceedings of the 57th Annual ACM Symposium on Theory of Computing},
  pages={36--44},
  year={2025}
}

@article{bernstein2025negative,
  title={{Negative-Weight Single-Source Shortest Paths in Near-Linear Time}},
  author={Bernstein, Aaron and Nanongkai, Danupon and Wulff-Nilsen, Christian},
  journal={Communications of the ACM},
  volume={68},
  number={2},
  pages={87--94},
  year={2025},
  publisher={ACM New York, NY, USA}
}

@article{chen2025maximum,
  title={{Maximum Flow and Minimum-Cost Flow in Almost-Linear Time}},
  author={Chen, Li and Kyng, Rasmus and Liu, Yang and Peng, Richard and Probst Gutenberg, Maximilian and Sachdeva, Sushant},
  journal={Journal of the ACM},
  volume={72},
  number={3},
  pages={1--103},
  year={2025},
  publisher={ACM New York, NY}
}

@misc{quanrud2025sparsification,
      title={{From Hop Reduction to Sparsification for Negative Length Shortest Paths}}, 
      author={Kent Quanrud and Navid Tajkhorshid},
      note={To appear in STOC 2026},
      year={2026} 
}

@misc{li2025shortcutting,
      title={{Shortcutting for Negative-Weight Shortest Path}}, 
      author={George Z. Li and Jason Li and Satish Rao and Junkai Zhang},
      note={To appear in STOC 2026},
      year={2026} 
}
\end{document}